\documentclass[sigconf]{aamas}  

\usepackage{url}
\usepackage{graphicx}
\usepackage{color}
\usepackage{amsmath}
\usepackage{amsthm}
\usepackage{dsfont}
\usepackage{amssymb}
\usepackage{algorithm}
\usepackage[noend]{algpseudocode}
\usepackage{subcaption}
\usepackage{tabularx}
\usepackage{enumerate}
\usepackage{cases}
\usepackage{multicol}
\usepackage{booktabs}
\usepackage{multirow}

\theoremstyle{definition}
\theoremstyle{remark}

\newtheorem{innercustomthm}{Theorem}
\newenvironment{customthm}[1]
	{\renewcommand\theinnercustomthm{#1}\innercustomthm}
	{\endinnercustomthm}

\DeclareMathOperator*{\argmin}{arg\,\min}

\renewcommand\footnotetextcopyrightpermission[1]{} 
\begin{document}

\title{Anytime Heuristic for Weighted Matching Through Altruism-Inspired Behavior}

\author{Panayiotis Danassis, Aris Filos-Ratsikas, Boi Faltings}
\affiliation{\institution{Artificial Intelligence Laboratory (LIA), \'Ecole Polytechnique F\'ed\'erale de Lausanne (EPFL)}}
\email{{panayiotis.danassis, aris.filosratsikas, boi.faltings}@epfl.ch}

\begin{abstract}
	We present a novel anytime heuristic (ALMA), inspired by the human principle of altruism, for solving the assignment problem. ALMA is decentralized, completely uncoupled, and requires no communication between the participants. We prove an upper bound on the convergence speed that is polynomial in the desired number of resources and competing agents per resource; crucially, in the realistic case where the aforementioned quantities are bounded independently of the total number of agents/resources, the convergence time remains \emph{constant} as the total problem size increases. 
	
	We have evaluated ALMA under three test cases: (i) an anti-coordination scenario where agents with similar preferences compete over the same set of actions, (ii) a resource allocation scenario in an urban environment, under a constant-time constraint, and finally, (iii) an on-line matching scenario using \emph{real} passenger-taxi data. In all of the cases, ALMA was able to reach high social welfare, while being orders of magnitude faster than the centralized, optimal algorithm. The latter allows our algorithm to scale to realistic scenarios with hundreds of thousands of agents, e.g., vehicle coordination in urban environments.
\end{abstract}


\maketitle

\section{Introduction} \label{Introduction}

One of the most relevant problems in multi-agent systems (MAS) is finding an optimal allocation between agents. This pertains to role allocation (e.g., team formation for autonomous robots \cite{GUNN201322}), task assignment (e.g., employees of a factory, taxi-passenger matching \cite{varakantham2012decision}), resource allocation (e.g., parking spaces and/or charging stations for autonomous vehicles \cite{geng2013new}), etc. What follows is \emph{applicable to any such scenario}, but for concreteness we will refer to the allocation of a set of resources to a set of agents, a setting known as the \emph{assignment problem}, one of the most fundamental combinatorial optimization problems \cite{munkres1957algorithms}.

When designing algorithms for assignment problems, a significant challenge emerges from the nature of real-world applications, which is often distributed and information-restrictive. For the former part, a variety of decentralized algorithms have been developed \cite{giordani2010distributed,7991447,zavlanos2008distributed,burger2012distributed}, all of which, though, require polynomial in the problem size number of messages. However, inter-agent interactions often repeat no more than a few hundreds of times. Moreover, sharing plans and preferences creates high overhead, and there is often a lack of responsiveness and/or communication between the participants \cite{AAAI10-adhoc}. Achieving fast convergence and high efficiency in such information-restrictive settings is extremely challenging. Yet, humans are able to routinely and robustly coordinate in similar everyday scenarios. One driving factor that facilitates human cooperation is the principle of \emph{altruism} \cite{doi:10.1162/003355302760193904,nowak2005evolution,gintis2000strong}. Inspired by human behavior, the proposed heuristic is modeled on the principle of altruism. This results to fast convergence to highly efficient allocations, without any communication between the agents.

A distinctive characteristic of ALMA is that agents make decisions locally, based on (i) the contest for resources that \emph{they} are interested in, (ii) the agents that are interested in the \emph{same} resources. If each agent is interested in only a \emph{subset} of the total resources, ALMA converges in time polynomial in the maximum size of the subsets; not the total number of resources. In particular, if the size of each subset is a constant fraction of the total number of resources, then the convergence time is \emph{constant}, in the sense that it does not grow with the problem size. The same is not true for other algorithms (e.g., the optimal centralized solution) which require time polynomial in the \emph{total} number of agents/resources, even if the aforementioned condition holds. The condition holds by default in many real-world applications; agents have only local knowledge of the world, there is typically a cost associated with acquiring a resource, or agents are simply only interested in resources in their vicinity (e.g., urban environments). This is important, as the proposed approach avoids  having to artificially split the problem in subproblems (e.g, by placing bounds or spatial constraints) and solve those separately, in order to make it tractable. Instead, ALMA utilizes a natural domain characteristic, instead of an artificial optimization technique (i.e., artificial bounds). Coupled to the convergence time, the decentralized nature of ALMA makes it applicable to large-scale, real-world applications (e.g., IoT devices, intelligent infrastructure, autonomous vehicles, etc.).

\subsection{Our Results}

Our main contributions in this paper are:

\textbf{(1)} We introduce a novel, anytime \textbf{AL}truistic \textbf{MA}tching heuristic (\textbf{ALMA}) for solving the assignment problem. ALMA is decentralized, completely uncoupled (i.e., each agent is only aware of his own history of action/reward pairs \cite{DBLP:journals/corr/Talebi13}), and requires no communication between the agents.

\textbf{(2)} We prove that if we bound the maximum number of resources an agent is interested in, and the maximum number of agents competing for a resource, the expected number of steps for any agent to converge is independent of the total problem size. Thus, we do not require to artificially split the problem, or similar techniques, to render it manageable.
	
\textbf{(3)} We provide a thorough empirical evaluation of ALMA on both synthetic and real data. In particular, we have evaluated ALMA under three test cases: (i) an anti-coordination scenario where agents with similar preferences compete over the same set of actions, (ii) a resource allocation scenario in an urban environment, under a constant-time constraint, and finally, (iii) an on-line matching scenario using \emph{real} passenger-taxi data. In all of the cases, ALMA achieves high social welfare (total satisfaction of the agents) as compared to the optimal solution, as well as various other algorithms.


\subsection{Related Work} \label{Related Work}

The assignment problem consists of finding a maximum weight matching in a weighted bipartite graph and it is one of the best-studied combinatorial optimization problems in the literature. The first polynomial time algorithm (with respect to the total number of nodes, and edges) was introduced by Jacobi in the 19th century \cite{borchardt1865investigando,ollivier2009looking}, and was succeeded by many classical algorithms \cite{munkres1957algorithms,Edmonds:1972:TIA:321694.321699,bertsekas1979distributed} with the \emph{Hungarian} algorithm of \cite{kuhn1955hungarian} being the most prominent one (see \cite{su2015algorithms} for an overview). The problem can also be solved via linear programming \cite{dantzig1990origins}, as its LP formulation relaxation admits integral optimal solutions \cite{Papadimitriou:1982:COA:31027}. In Section \ref{testcase 3}, we will apply ALMA on a non-bipartite setting, which corresponds to the more general \emph{maximum weight matching} problem on general graphs. To compute the optimal in this case, we will use the \emph{blossom algorithm} of \cite{edmonds1965maximum} (see \cite{lovasz2009matching}).

In reality, a centralized coordinator is not always available, and if so, it has to know the utilities of all the participants, which is often not feasible. In the literature of the assignment problem, there also exist several decentralized algorithms (e.g., \cite{giordani2010distributed,7991447,zavlanos2008distributed,burger2012distributed} which are the decentralized versions of the aforementioned well-known centralized algorithms). However, these algorithms require polynomial computational time and polynomial number of messages (such as cost matrices \cite{7991447}, pricing information \cite{zavlanos2008distributed}, or a basis of the LP \cite{burger2012distributed}, etc.). Yet, agent interactions often repeat no more than a few hundreds of times. To the best of our knowledge, a decentralized algorithm that requires no message exchange (i.e., no communication network) between the participants, and achieves high efficiency, like ALMA does, has not appeared in the literature before. Let us stress the importance of such a heuristic: as autonomous agents proliferate, and their number and diversity continue to rise, differences between the agents in terms of origin, communication protocols, or the existence of sub-optimal, legacy agents will bring forth the need to collaborate without any form of explicit communication \cite{AAAI10-adhoc}. Finally, inter-agent communication creates high overhead as well.

ALMA is inspired by the decentralized allocation algorithm of \cite{DANASSIS:2019}. Using such a simple learning rule which only requires environmental feedback, allows our approach to scale to hundreds of thousands of agents. Moreover, it does not require global knowledge of utilities; only local knowledge of personal utilities (in fact, we require knowledge of pairwise differences which are far easier to estimate).

\section{Altruistic Matching Heuristic} \label{Proposed Approach}

In this section, we define ALMA and prove its convergence properties. We begin with the definition of the assignment problem, and its interpretation in our setting.

\subsection{The Assignment Problem}

The assignment problem consists of finding a maximum weight matching in a weighted bipartite graph, $\mathcal{G} = \left\{ \mathcal{N} \cup \mathcal{R}, \mathcal{E} \right\}$. In the studied scenario, $\mathcal{N} = \{1, \dots, N\}$ agents compete to acquire $\mathcal{R} = \{1, \dots, R\}$ resources. We assume that each agent $n$ is interested in a subset of the total resources, i.e., $\mathcal{R}^n \subset \mathcal{R}$. The weight of an edge $(n, r) \in \mathcal{E}$ represents the utility ($u_n(r)$) agent $n$ receives by acquiring resource $r$. Each agent can acquire at most one resource, and each resource can be assigned to at most one agent. The goal is to maximize the social welfare (sum of utilities), i.e., 


\begin{equation}\label{eq:optimization}
	\begin{aligned}
		\max_{\mathbf{x} \geq 0}
		& \sum_{(n,r) \in \mathcal{E}} u_{n, r} x_{n, r} \\
		\text{subject to}
		& \sum_{r | (n,r) \in \mathcal{E}} x_{n, r} = 1, \forall n \in \mathcal{N} \\
		& \sum_{n | (n,r) \in \mathcal{E}} x_{n, r} = 1, \forall r \in \mathcal{R}
	\end{aligned}
\end{equation}

\subsection{Learning Rule} \label{Learning Rule}

This section describes the proposed heuristic (\textbf{ALMA}: \textbf{AL}truistic \textbf{MA}tching heuristic) for weighted matching. We make the following two assumptions: First, we assume (possibly noisy) knowledge of personal preferences by each agent. Second, we assume that agents can observe feedback from their environment. This is used to inform collisions and detect free resources. It could be achieved by the use of visual, auditory, olfactory sensors etc., or by any other means of feedback from the resource (e.g., by sending an occupancy message). Note here that these messages would be between the requesting agent and the resource, not between the participating agents themselves, and that it suffices to send only 1 bit of information (e.g., 0, 1 for occupied / free respectively). 

ALMA learns the right action through repeated trials as follows. Each agent sorts his available resources (possibly $\mathcal{R}^n \subseteq \mathcal{R}$) in decreasing order of utility ($r_1, r_2, \dots,$ $r_i, r_{i + 1}, \dots, r_{R^n}$). The set of available actions is denoted as $\mathcal{A} = \{Y, A_{r_1}, \dots, A_{r_{R^n}}\}$, where $Y$ refers to yielding, and $A_r$ refers to accessing resource $r$. Each agent has a strategy ($g_n$) that points to a resource and it is initialized to the most preferred one. As long as an agent has not acquired a resource yet, at every time-step, there are two possible scenarios. If $g_n = A_r$ (strategy points to resource $r$), then agent $n$ attempts to acquire that resource. If there is a collision, the colliding parties back-off with some probability. Otherwise, if $g_n = Y$, the agent choses a resource $r$ for monitoring. If the resource is free, he sets $g_n \leftarrow A_r$. Alg. \ref{algo: learning rule} presents the pseudo-code of ALMA, which is followed by every agent individually. The back-off probability and the next resource to monitor are computed individually and locally based on the current resource and each agent's utilities, as will be explained in the following section. Finally, note that if the available resources change over time, the agents simply need to sort again the currently available ones.

\subsection{Back-off Probability \& Resource Selection} \label{Back-off Probability & Resource Selection}

Let $\mathcal{R}$ be totally ordered in decreasing utility under $\prec_n$, $\forall n \in \mathcal{N}$. If more than one agent compete for resource $r_i$ (step 4 of Alg. \ref{algo: learning rule}), each of them will back-off with probability that depends on their utility loss of switching to their respective remaining resources. The loss is given by Eq. \ref{Eq: loss}.

\begin{equation} \label{Eq: loss}
	loss_n^i = \frac{\underset{j = i + 1}{\overset{k}{\sum}} u_n(r_i) - u_n(r_j)}{k - i}
\end{equation}

\noindent
where $k \in \{i + 1, \dots, R^n\}$ denotes the number of remaining resources to be considered. For $k = i + 1$, the formula only takes into account the utility loss of switching to the immediate next best resource, while for $k = R^n$ it takes into account the average utility loss of switching to all of the remaining resources. In the remainder of the paper we assume $k = i + 1$, i.e., $loss_n^i = u_n(r_i) - u_n(r_{i + 1})$. The actual back-off probability can be computed with any monotonically decreasing function $f$ on $loss_n$, i.e., $P_n(r_i, \prec_n) = f_n(loss_n^i)$. In the evaluation section, we have used two such functions, a linear (Eq. \ref{Eq: linear}), and the logistic function (Eq. \ref{Eq: logistic}). The parameter $\epsilon$ places a threshold on the minimum / maximum back-off probability for the linear curve, while $\gamma$ determines the steepness of the logistic curve.

\begin{equation} \label{Eq: linear}
	f(loss) =
	\begin{cases}
		1 - \epsilon, & \text{ if } loss \leq \epsilon \\
		\epsilon, & \text{ if } 1 - loss \leq \epsilon \\
		1 - loss, & \text{ otherwise}
	\end{cases}
\end{equation}

\begin{equation} \label{Eq: logistic}
	f(loss) = \frac{1}{1 + e^{-\gamma(0.5 - loss)}}
\end{equation}

Using the aforedescribed rule, agents that do not have good alternatives will be less likely to back-off and vice versa. The ones that do back-off select an alternative resource and examine its availability. The resource selection is performed in sequential order, i.e., $S_n(r_{\text{prev}}, \prec_n) = r_{\text{prev} + 1}$, where $r_{\text{prev}}$ denotes the resource selected by that agent in the previous round. We also examined the possibility of using a weighted or uniformly at random selection, but achieved inferior results.

\subsection{Altruism-Inspired Behavior} \label{Inspiration}

ALMA is inspired by the human principle of altruism. We would expect an altruistic person to give up a resource either to someone who values it more, if that resulted in an improvement of the well-being of society \cite{doi:10.1162/003355302760193904}, or simply to be nice to others \cite{simon2016existence}. Such behavior is especially common in situations where the backing-off subject has equally good alternatives. For example, in human pick-up teams, each player typically attempts to fill his most preferred position. If there is a collision, a colliding player might back-off because his teammate is more competent in that role, or because he has an equally good alternative, or simply to be polite; the player backs-off now and assumes that role at some future game. From an alternative viewpoint, following such an altruistic convention leads to a faster convergence which outweighs the loss in utility. Such conventions allow humans to routinely and robustly coordinate in large scale and under dynamic and unpredictable demand. Behavioral conventions are a fundamental part of human societies \cite{lewis2008convention}, yet they have not appeared meaningfully in empirical modeling of multi-agent systems. Inspired by human behavior, ALMA attempts to reproduce these simple rules in an artificial setting.

\begin{algorithm}[!t]
	\caption{ALMA: Altruistic Matching Heuristic.} \label{algo: learning rule}
	\begin{algorithmic}[1]
		\Require Sort resources ($\mathcal{R}^n \subseteq \mathcal{R}$) in decreasing order of utility $r_1, r_2, \dots, r_i, r_{i + 1}, \dots, r_{R^n}$.
		\Require Initialize $g_n \leftarrow A_{r_1}$, and $r_{\text{prev}} \leftarrow r_1$.

		\Procedure{ALMA}{}
			\If{ $g_n = A_r$}
				\State Agent $n$ attempts to acquire resource $r$. Set $r_{\text{prev}} \leftarrow r$.
				\If{Collision($r$)}
				\State back-off (set $g_n \leftarrow Y$) with probability $P_n(r, \prec_n)$.
				\EndIf
			\Else { ($g_n = Y$)}
				\State Agent $n$ monitors $r \leftarrow S_n(r_{\text{prev}}, \prec_n)$. Set $r_{\text{prev}} \leftarrow r$.
				\If{Free($r$)} set $g_n \leftarrow A_r$.
				\EndIf
			\EndIf    
		\EndProcedure
	\end{algorithmic}
\end{algorithm}

\subsection{Convergence} \label{Convergecne}

Agents who have not acquired a resource ($g_n = Y$) will not claim an occupied one. Additionally, every time a collision happens, there is a positive probability that some agents will back-off. As a result, the system will converge. The following theorem proves that the expected convergence time is logarithmic in the number of agents $N$ and quadratic in the number of resources $R$.

\begin{theorem} \label{Th: convergence system}
For $N$ agents and $R$ resources, the expected number of steps until the system of agents following Alg. \ref{algo: learning rule} converges to a complete matching is bounded by (\ref{Eq: convergence bound system}), where $p^* = f(loss^*)$, and $loss^*$ is given by Eq. \ref{Eq: loss* system}.

\begin{equation} \label{Eq: convergence bound system}
	\mathcal{O}\left( R \frac{2 - p^*}{2 (1 - p^*)} \left(\frac{1}{p^*} \log N + R \right) \right) 
\end{equation}

\begin{equation} \label{Eq: loss* system}
	loss^* = \underset{loss_n^r}{\argmin} \left( \underset{r \in \mathcal{R}, n \in \mathcal{N}}{\min}(loss_n^r), 1 - \underset{r \in \mathcal{R}, n \in \mathcal{N}}{\max}(loss_n^r) \right)
\end{equation}
\end{theorem}

\begin{proof}
	To improve readability, we will only provide a sketch of the proof. Please see the appendix for the complete version. The proof is based on \cite{cigler2011reaching,DANASSIS:2019}.

	We first assume that every agent, on every collision, backs-off with the same constant probability $p$. We start with the case of having $N$ agents competing for 1 resource and model our system as a discrete time Markov chain. Intuitively, this Markov chain describes the number of individuals in a decreasing population, but with two caveats: the goal (absorbing state) is to reach a point where only one individual remains, and if we reach zero, we restart. We prove that the expected number of steps until we reach a state where either 1 or 0 agents compete for that resource is $\mathcal{O}\left( \frac{1}{p} \log N \right)$. Moreover, we prove that with high probability, $\Omega\left( \frac{2(1 - p)}{2 - p} \right)$, only 1 agent will remain (contrary to reaching 0 and restarting the process of claiming the resource), no matter the initial number of agents. Having proven that, we move to the general case of $N$ agents competing for $R$ resources.

	At any time, at most $N$ agents can compete for each resource. We call this period a round. During a round, the number of agents competing for a specific resource monotonically decreases, since that resource is perceived as occupied by non-competing agents. Let the round end when either 1 or 0 agents compete for the resource. This will require $\mathcal{O}\left( \frac{1}{p} \log N \right)$ steps. If all agents backed-off, it will take on average $R$ steps until at least one of them finds a free resource. We call this period a break. In the worst case, the system will oscillate between a round and a break. According to the above, one oscillation requires in expectation $\mathcal{O}\left( \frac{1}{p} \log N  + R\right)$ steps. If $R = 1$, as mentioned in the previous paragraph, in expectation there will be $\frac{2 - p}{2 (1 - p)}$ oscillations. For $R > 1$ the expected number of oscillations is bounded by $\mathcal{O}\left( R \frac{2 - p}{2 (1 - p)} \right)$. Thus, we conclude that if all the agents back-off with the same constant probability $p$, the expected number of steps until the system converges to a complete matching is $\mathcal{O}\left( R \frac{2 - p}{2 (1 - p)} \left(\frac{1}{p} \log N + R \right) \right)$.

	Next, we drop the constant probability assumption. Intuitively, the worst case scenario corresponds to either all agents having a small back-off probability, thus they keep on competing for the same resource, or all of them having a high back-off probability, thus the process will keep on restarting. These two scenarios correspond to the inner ($\frac{1}{p}$) and outer ($\frac{2 - p}{2 (1 - p)}$) probability terms of bound (\ref{Eq: convergence bound system}) respectively. Let $p^* = f(loss^*)$ be the worst between the smallest or highest back-off probability any agent $n \in \mathcal{N}$ can exhibit, i.e., having $loss^*$ given by Eq. \ref{Eq: loss* system}. Using $p^*$ instead of the constant $p$, we bound the expected convergence time according to bound (\ref{Eq: convergence bound system}).
\end{proof}

It is worth noting that the back-off probability $p^*$ in bound (\ref{Eq: convergence bound system}) does not significantly affect the convergence time. For example, using Eq. \ref{Eq: linear} with a quite small $\epsilon = 0.01$, the resulting quantities would be at most $100 R \log N$, and $50 R^2$. Most importantly, though, this is a rather loose bound (e.g., agents would rarely back-off with probabilities as extreme as $p^*$).

Apart from the convergence of the whole system, we are interested in the expected number of steps any individual agent would require in order to acquire a resource. In real-world scenarios, there is typically a cost associated with acquiring a resource. For example, a taxi driver would not be willing to drive to the other end of the city to pick up a low fare passenger. As a result, each agent is typically interested in a subset of the total resources, i.e., $\mathcal{R}^n \subset \mathcal{R}$, thus at each resource there is a bounded number of competing agents. Let $R^n$ denote the maximum number of resources agent $n$ is interested in, and $N^r$ denote the maximum number of agents competing for resource $r$. By bounding these two quantities (i.e., we consider $R^n$ and $N^r$ to be constant functions of $N$, $R$), Corollary \ref{Th: convergence agent} proves that the expected number of steps any individual agent requires in order to claim a resource is independent of the total problem size (i.e., $N$, and $R$), or, in other words, that the convergence time is \emph{constant} in these quantities.

\begin{corollary} \label{Th: convergence agent}
Let $R^n = |\mathcal{R}^n|$, such that $\forall r \in \mathcal{R}^n : u_n(r) > 0$, and $N^r = |\mathcal{N}^r|$ , such that $\forall n \in \mathcal{N}^r : u_n(r) > 0$. The expected number of steps until an agent $n \in \mathcal{N}$ following Alg. \ref{algo: learning rule} successfully acquires a resource is bounded by (\ref{Eq: convergence bound agent}), where $p_n^* = f(loss^\star)$ and $loss^\star$ is given by Eq. \ref{Eq: loss* agent}, independent of the total problem size $N$, $R$.

\begin{equation} \label{Eq: convergence bound agent}
	\mathcal{O}\left( \underset{n' \in \cup_{r \in \mathcal{R}^n} \mathcal{N}^r}{\max R^{n'}} \frac{2 - p_n^*}{2 (1 - p_n^*)} \left(\frac{1}{p_n^*} \log ( \underset{r \in \mathcal{R}^n}{\max} N^r ) + \underset{n' \in \cup_{r \in \mathcal{R}^n} \mathcal{N}^r}{\max R^{n'}} \right) \right) 
\end{equation}

\begin{equation} \label{Eq: loss* agent}
	loss^\star = \underset{loss_n^r}{\argmin} \left( \underset{r \in \mathcal{R}^n, n \in \mathcal{N}^r}{\min}(loss_n^r), 1 - \underset{r \in \mathcal{R}^n, n \in \mathcal{N}^r}{\max}(loss_n^r) \right)
\end{equation}
\end{corollary}

\begin{proof}
	The expected number of steps until an agent $n \in \mathcal{N}$ successfully acquires a resource is upper bounded by the total convergence time of the sub-system he belongs to, i.e., the sub-system consisting of the sets of $\mathcal{R}^n$ resources and $\cup_{r \in \mathcal{R}^n} \mathcal{N}^r$ agents. In such scenario, at most $\max_{r \in \mathcal{R}^n} N^r$ agents can compete for any resource. Using Theorem \ref{Th: convergence system} for $\max_{r \in \mathcal{R}^n} N^r$ agents, $\max_{n' \in \cup_{r \in \mathcal{R}^n} \mathcal{N}^r}R^{n'}$ resources, and worst-case $loss^\star$ given by any agent in $\cup_{r \in \mathcal{R}^n} \mathcal{N}^r$ (i.e., Eq \ref{Eq: loss* agent}) results in the desired bound. Note that agents do not compete for already claimed resources (step 8 of Alg. \ref{algo: learning rule}), thus the convergence of an agent does not require the convergence of agents of overlapping sub-systems.
\end{proof}

\begin{figure*}[t!]
	\centering
	\begin{subfigure}[t]{0.5\textwidth}
		\centering
		\includegraphics[width = 1 \linewidth]{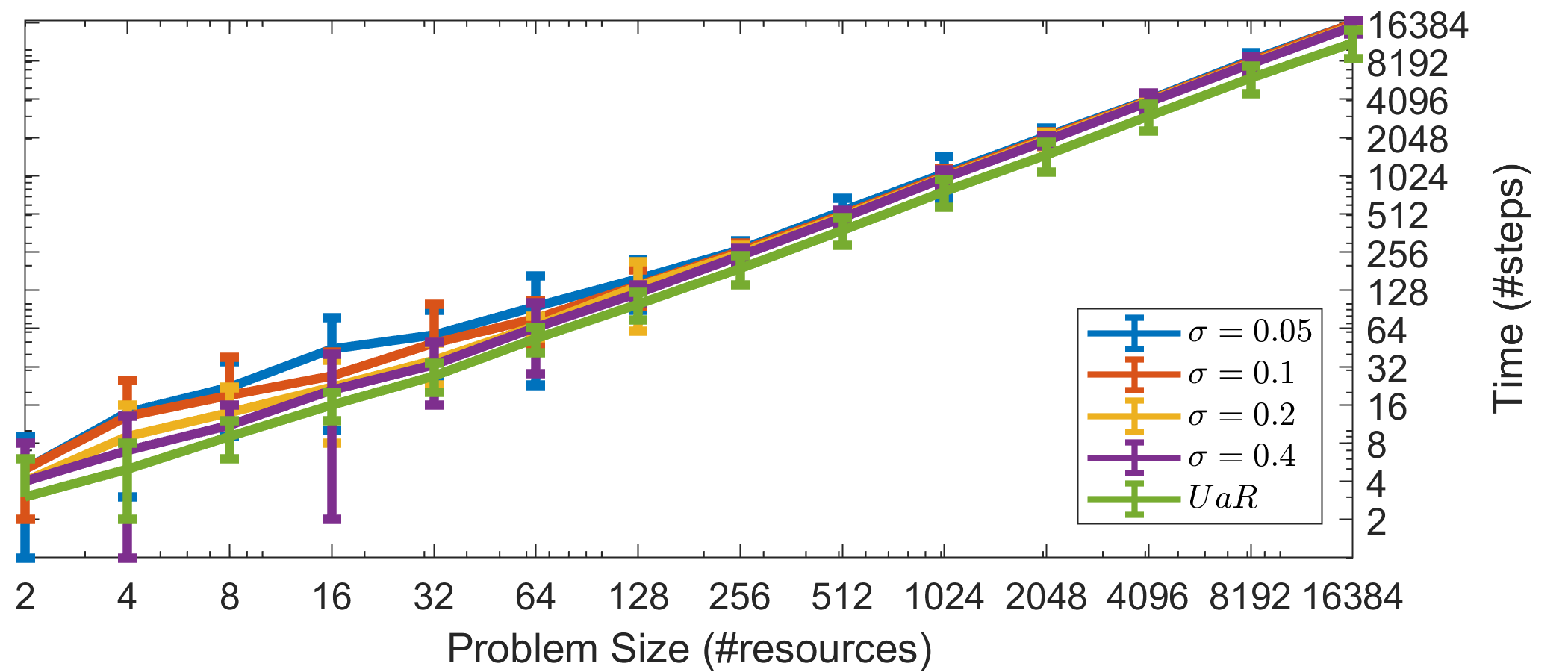}
		\caption{}
		\label{fig: testcase1_gaussian_uar_steps}
	\end{subfigure}%
	~ 
	\begin{subfigure}[t]{0.5\textwidth}
		\centering
		\includegraphics[width = 1 \linewidth]{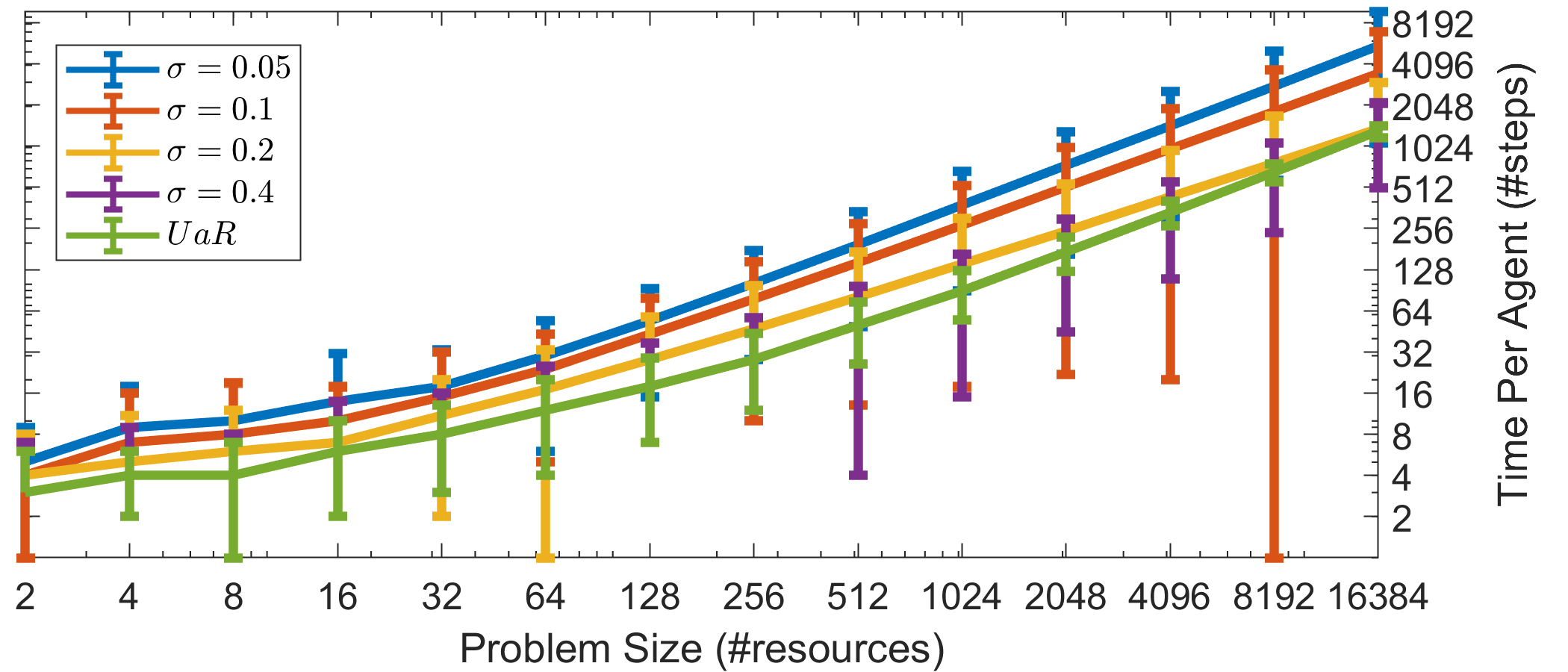}
		\caption{}
		\label{fig: testcase1_gaussian_uar_stepsPerAgent}
	\end{subfigure}
	~ \\
	\begin{subfigure}[t]{0.5\textwidth}
		\centering
		\includegraphics[width = 1 \linewidth]{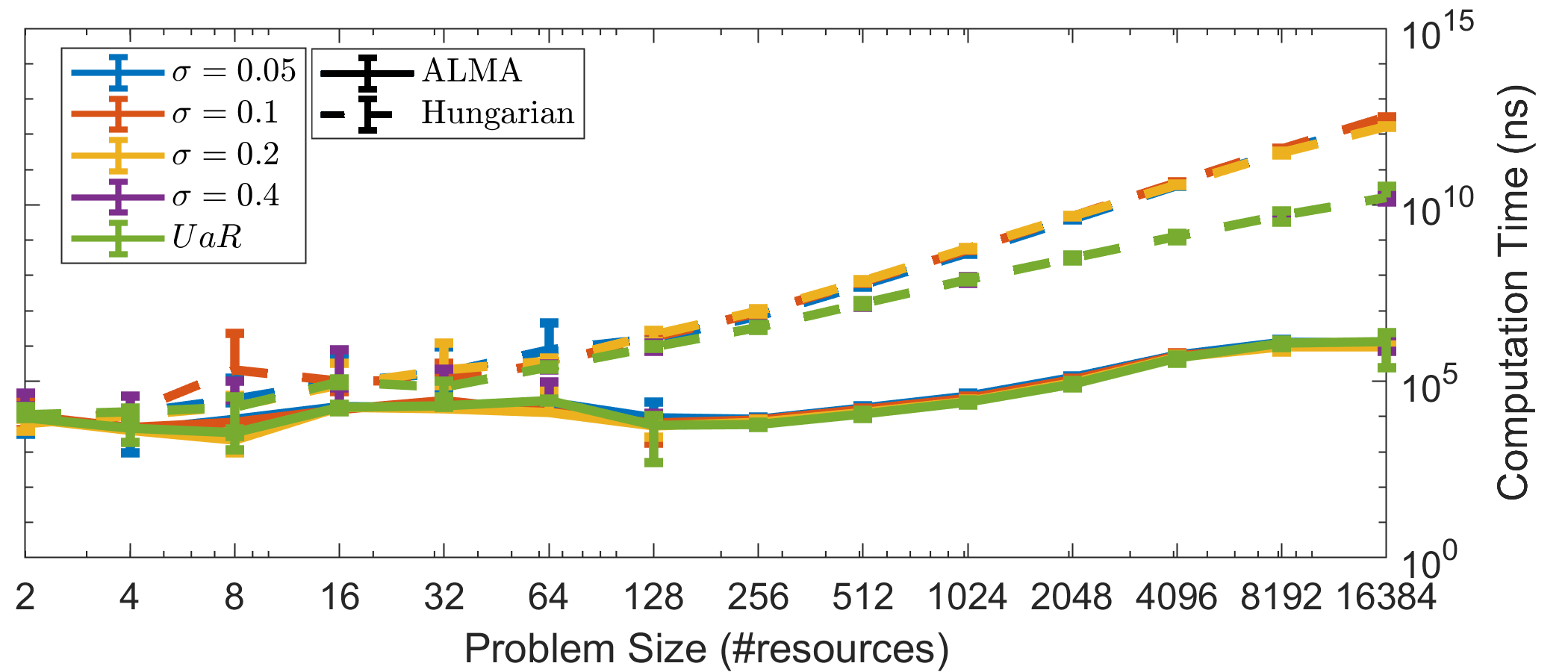}
		\caption{}
		\label{fig: testcase1_gaussian_uar_computationTime}
	\end{subfigure}%
	~ 
	\begin{subfigure}[t]{0.5\textwidth}
		\centering
		\includegraphics[width = 1 \linewidth]{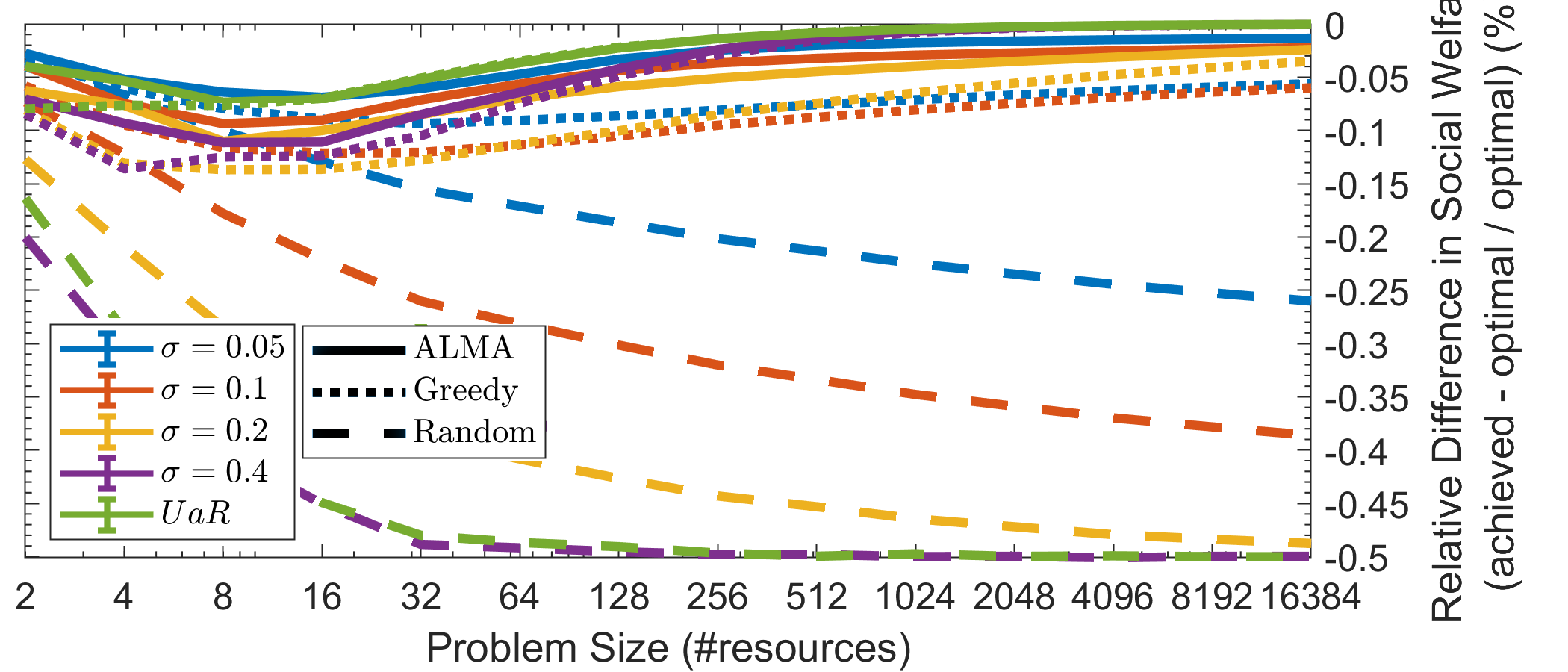}
		\caption{}
		\label{fig: testcase1_gaussian_uar_cumulativeRegret}
	\end{subfigure}
	\caption{From left to right, top to bottom: (\ref{fig: testcase1_gaussian_uar_steps}) Total convergence time (\#steps), (\ref{fig: testcase1_gaussian_uar_stepsPerAgent}) average time (\#steps) for an individual agent to successfully acquire a resource, (\ref{fig: testcase1_gaussian_uar_computationTime}) Computation time (ns), and (\ref{fig: testcase1_gaussian_uar_cumulativeRegret}) Relative difference in social welfare (\%), for increasing number of resources, and $N = R$. Fig. \ref{fig: testcase1_gaussian_uar_steps}, \ref{fig: testcase1_gaussian_uar_computationTime}, and \ref{fig: testcase1_gaussian_uar_stepsPerAgent} are in double log. scale, while Fig. \ref{fig: testcase1_gaussian_uar_cumulativeRegret} is in single log. scale.}
	\label{fig: testcase 1}
\end{figure*}

\section{Evaluation} \label{Evaluation}

In this section we evaluate ALMA under various test cases. For the first two, we focus on convergence time and relative difference in social welfare (SW), i.e., $(achieved - optimal) / optimal$. In every reported metric, except for the social welfare, we report the average value out of 128 runs of the same problem instance. Error bars represent one standard deviation (SD) of uncertainty. As a measure of social welfare, we report the cumulative regret of the aforementioned 128 runs, i.e., for $i \in [1, 128]$ runs, the reported relative difference in social welfare is $(\sum_i achieved - \sum_i optimal) / \sum_i optimal$. This was done to improve visualization of the results in smaller problem sizes, where really small differences result in high SD bars (e.g., if $achieved$ = $1\times10^{-5}$, and $optimal$ = $2\times10^{-5}$, the relative difference would be $-50\%$ for practically the same matching). The optimal matchings were computed using the Hungarian algorithm \footnote{We used Kevin L. Stern's $\mathcal{O}(N^3)$ implementation: \url{https://github.com/KevinStern/}.}. The third test case is an on-line setting, thus we report the achieved SW (not the relative difference to the optimal), and the empirical competitive ratio (average out of 128 runs, as before). All the simulations were run on 2x Intel Xeon E5-2680 with 256 GB RAM. In Section \ref{testcase 1} we use the logistic function (Eq. \ref{Eq: logistic}) with $\gamma = 2$, while in Sections \ref{testcase 2} \& \ref{testcase 3} we use the linear function (Eq. \ref{Eq: linear}) with $\epsilon = 0.1$.

It is important to stress that our goal is not to improve the convergence speed of a centralized, or decentralized algorithm. Rather, the computation time comparisons of Sections \ref{testcase 1} \& \ref{testcase 2} are meant to ground the actual speed of ALMA, and argue in favor of its applicability on large-scale, real-world scenarios. Given the nature of the problem, we elected to use a specialized algorithm to compute the optimal solution, rather than a general LP-based technique (e.g., the Simplex method). Specifically, we opted to use the Hungarian algorithm which, first, has proven polynomial worse case bound, and second, as our simulations will demonstrate, can handle sufficiently large problems.

\subsection{Test Case \#1: Uniform, and Noisy Common Preferences} \label{testcase 1}

\subsubsection{Setting}

As a first evaluation test case, we cover the extreme scenarios. The first pertains to an anti-coordination scenario in which agents with similar preferences compete over the same set of actions \cite{SSS1817485}. For example, autonomous vehicles would prefer the least congested route, bidding agents participating in multiple auctions would prefer the ones with the smallest number of participants, etc. We call this scenario `noisy common preferences' and model the utilities as follows: $\forall n, n' \in \mathcal{N}, |u_n(r) - u_{n'}(r)| \leq \text{noise}$, where the noise is sampled from a zero-mean Gaussian distribution, i.e., $\text{noise} \sim \mathcal{N}(0, \sigma^2)$ \footnote{Similar results achieved using uniform noise, i.e., $\sim \mathcal{U}(-\nu, \nu)$.}. In the second scenario the utilities are initialized uniformly at random ($UaR$) for each agent and resource.

\subsubsection{Convergence Time}

Starting with Fig. \ref{fig: testcase1_gaussian_uar_steps}, we can see that the convergence time for the system of agents following Alg. \ref{algo: learning rule} is linear to the number of resources $R$. From the perspective of a single agent, Fig. \ref{fig: testcase1_gaussian_uar_stepsPerAgent} shows that on average he will successfully acquire a resource significantly ($> 2\times$) faster than the total convergence time. This suggest that there is a small number of agents which take longer in claiming a resource and which in turn delay the system's convergence. We will exploit this property in the next section to present the anytime property of ALMA. Fig. \ref{fig: testcase1_gaussian_uar_computationTime} shows that ALMA requires approximately 4 to 6 orders of magnitude less computation time than the centralized Hungarian algorithm. Furthermore, ALMA seems to scale more gracefully, an important property for real world applications. Note also that in real-world applications we would have to take into account communication time, communication reliability protocols, etc., which create additional overhead for the Hungarian or any other algorithm for the assignment problem.

\subsubsection{Efficiency}

The relative difference in social welfare (Fig. \ref{fig: testcase1_gaussian_uar_cumulativeRegret}) reaches asymptotically zero as $R$ increases. For a small number of resources, ALMA achieves the worst social welfare, approximately $11\%$ worse than the optimal. Intuitively this is because when we have a small number of choices, a single wrong matching can have a significant impact to the final social welfare, while as the number of resources grow, the impact of an erroneous matching is mitigated. For 16384 resources we lose less than $2.5\%$ of the optimal. As a reference, Fig. \ref{fig: testcase1_gaussian_uar_cumulativeRegret} depicts the centralized greedy, and the random solutions as well. The greedy solution goes through the participating agents randomly, and assigns them their most preferred unassigned resource. In this scenario, the random solution loses up to $50\%$ of the optimal SW, while the greedy solution achieves similar results to ALMA, especially in high noise situations. This is to be expected, since first, all agents are interested in all the resources, and second, as the noise increases, the agents' preferences become more distinguishable, more diverse. ALMA is of a greedy nature as well, albeit it utilizes a more intelligent backing-off scheme. Contrary to that, the greedy solution does not take into account the utilities between agents, thus there are scenarios where ALMA would significantly outperform the greedy (e.g., see Section \ref{testcase 3}). Finally, recall that ALMA operates in a significantly harder domain with no communication, limited feedback, and time constraints. In contrast, the greedy method requires either a central coordinator or message exchange (to communicate users' preferences and resolve collisions).

\begin{figure*}[p!]
	\centering
	\begin{subfigure}[t]{0.5\textwidth}
		\centering
		\includegraphics[width = 1 \linewidth]{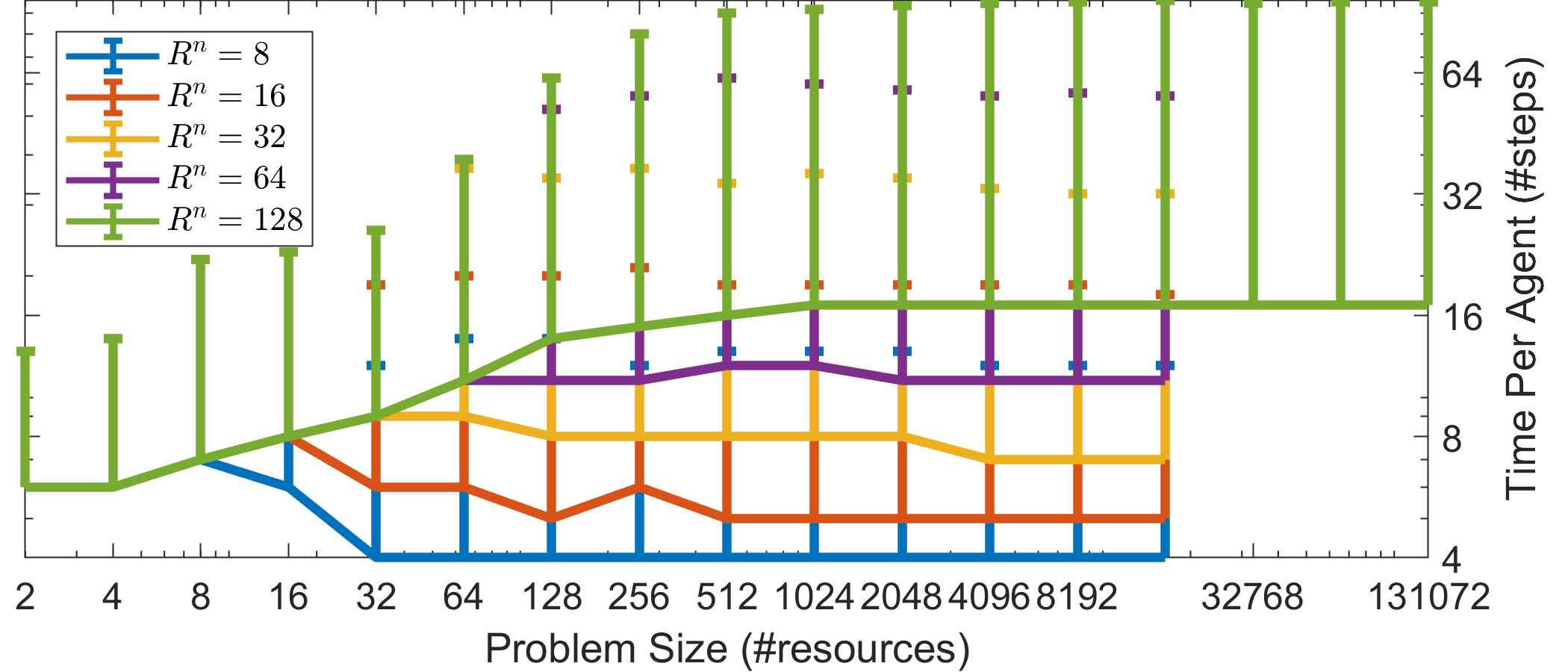}
		\caption{}
		\label{fig: testcase2_mapCutoffBounded_gridSpots_4_mapCutoffDistance_99999999_stepsPerAgent}
	\end{subfigure}%
	~ 
	\begin{subfigure}[t]{0.5\textwidth}
		\centering
		\includegraphics[width = 1 \linewidth]{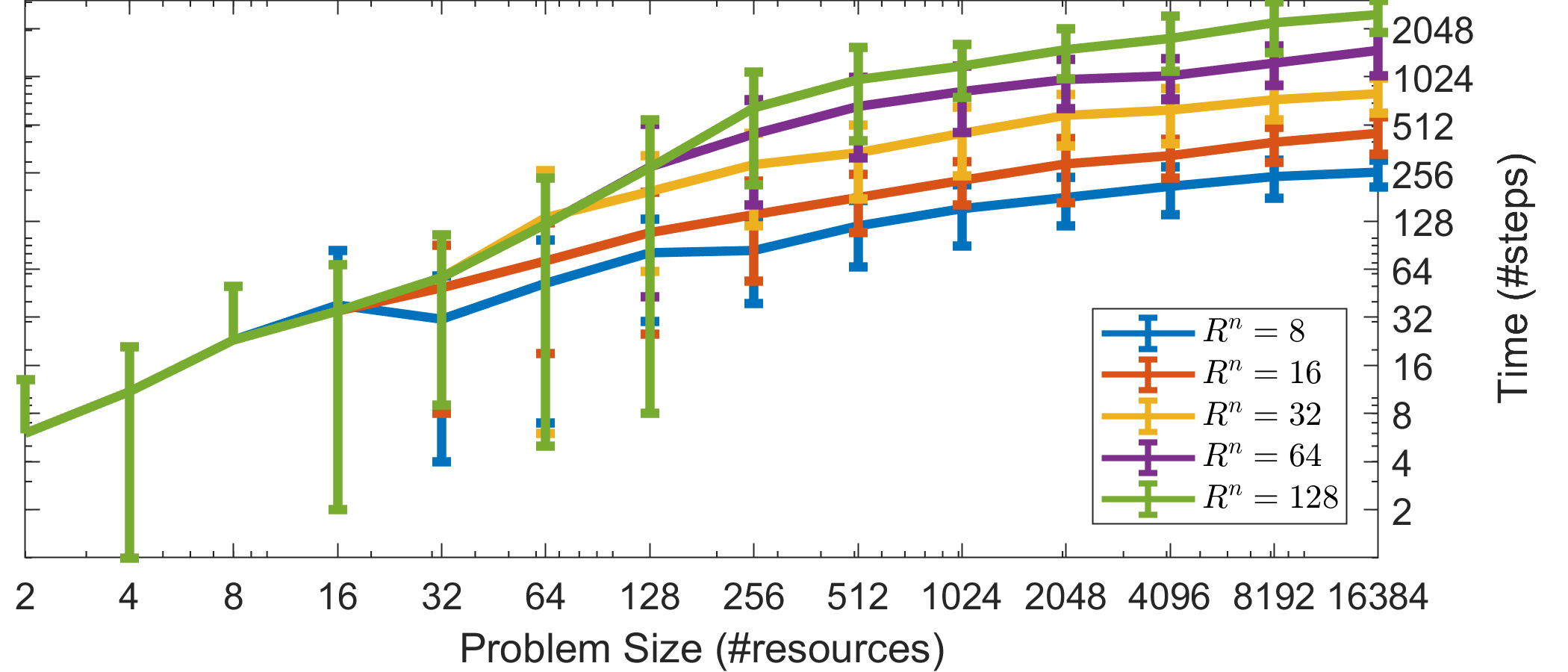}
		\caption{}
		\label{fig: testcase2_mapCutoffBounded_gridSpots_4_mapCutoffDistance_99999999_steps}
	\end{subfigure}
	~ \\
	\begin{subfigure}[t]{0.5\textwidth}
		\centering
		\includegraphics[width = 1 \linewidth]{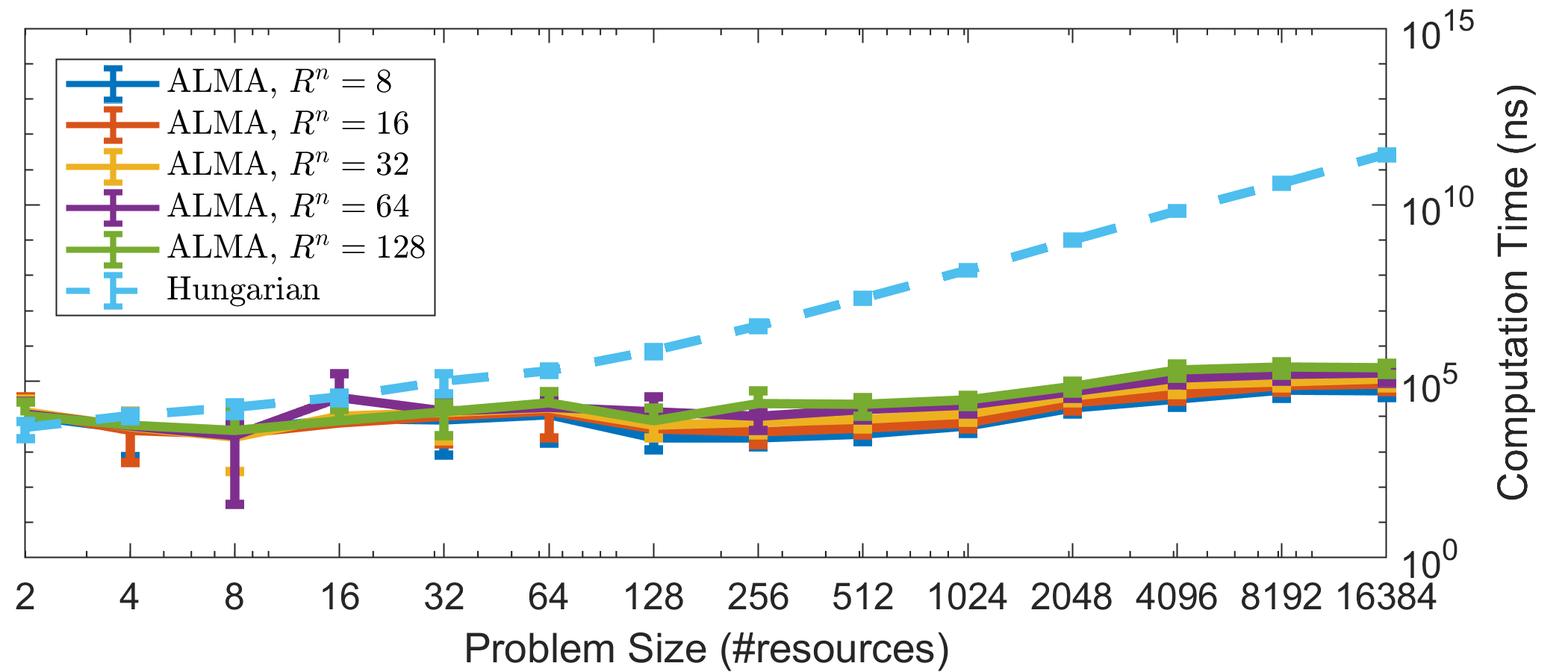}
		\caption{}
		\label{fig: testcase2_mapCutoffBounded_gridSpots_4_mapCutoffDistance_99999999_computationTime}
	\end{subfigure}%
	~ 
	\begin{subfigure}[t]{0.5\textwidth}
		\centering
		\includegraphics[width = 1 \linewidth]{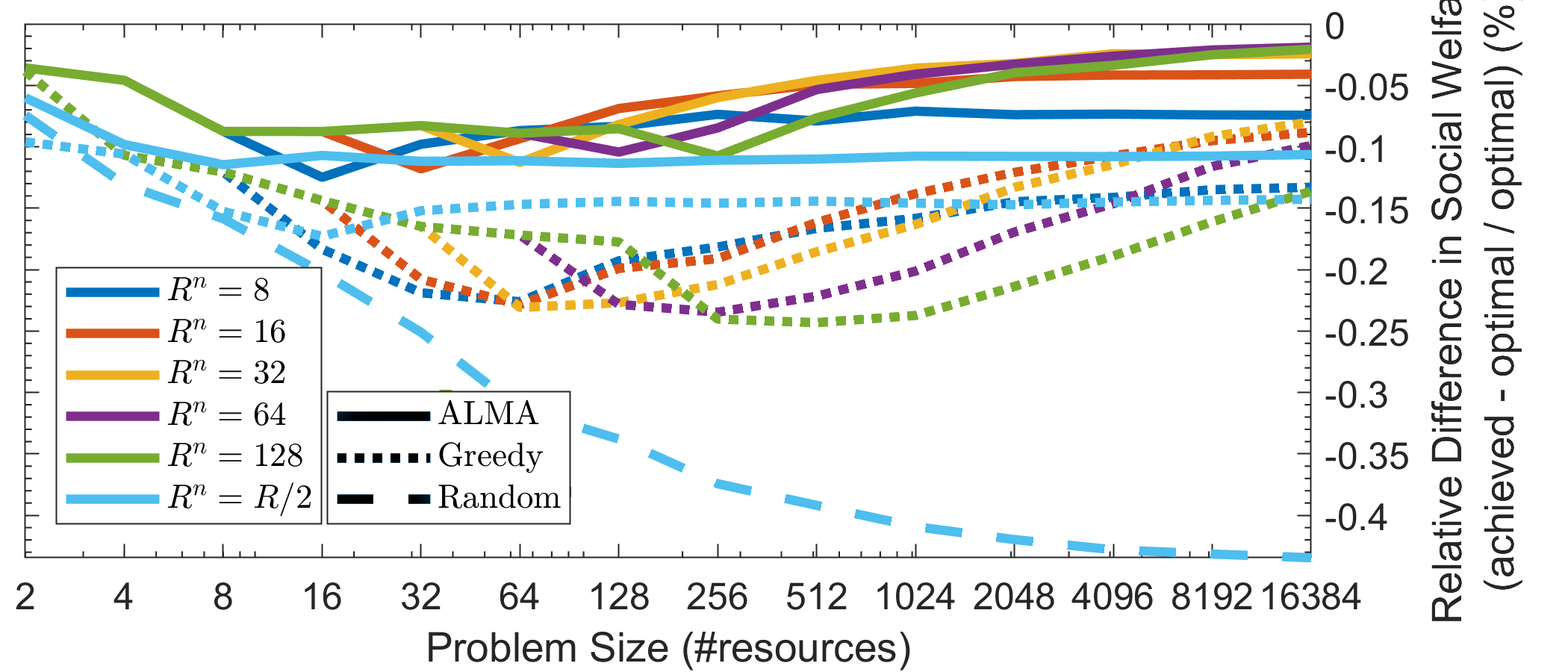}
		\caption{}
		\label{fig: testcase2_mapCutoffBounded_gridSpots_4_mapCutoffDistance_99999999_cumulativeRegret}
	\end{subfigure}
	~ \\
	\begin{subfigure}[t]{0.5\textwidth}
		\centering
		\includegraphics[width = 1 \linewidth]{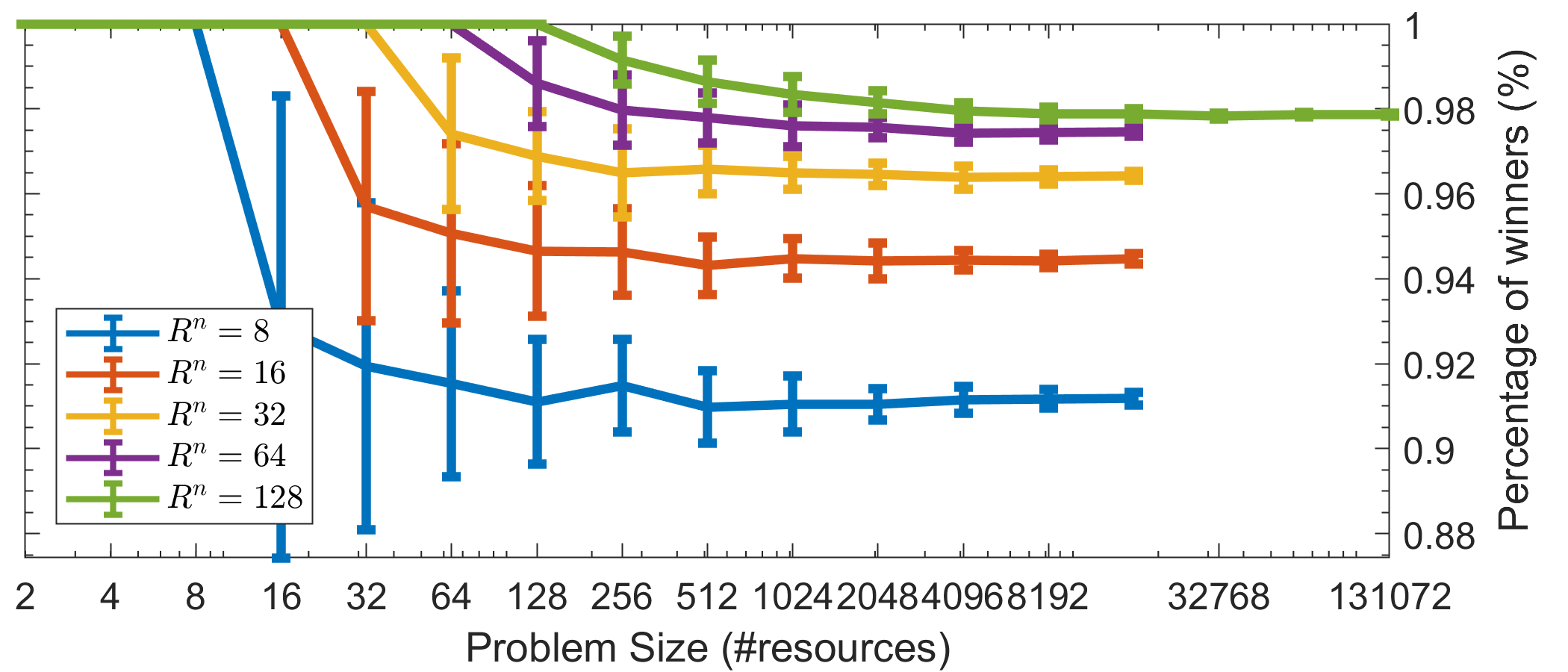}
		\caption{}
		\label{fig: testcase2_mapCutoffBounded_gridSpots_4_mapCutoffDistance_99999999_percentageOfWinners}
	\end{subfigure}%
	~
	\begin{subfigure}[t]{0.5\textwidth}
		\centering
		\includegraphics[width = 1 \linewidth]{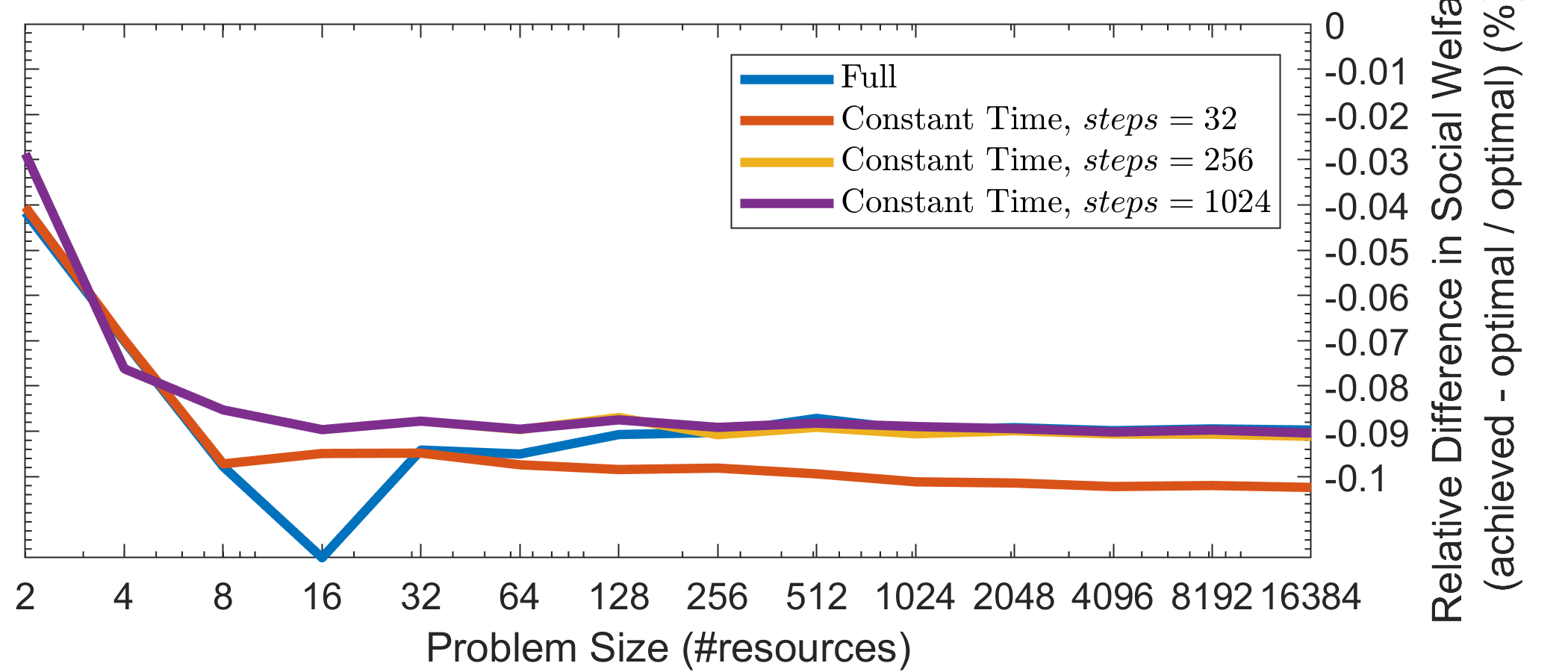}
		\caption{}
		\label{fig: testcase2_gridSpots_4_distance_025_cumulativeRegret}
	\end{subfigure}
	~ \\
	\begin{subfigure}[t]{0.5\textwidth}
		\centering
		\includegraphics[width = 1 \linewidth]{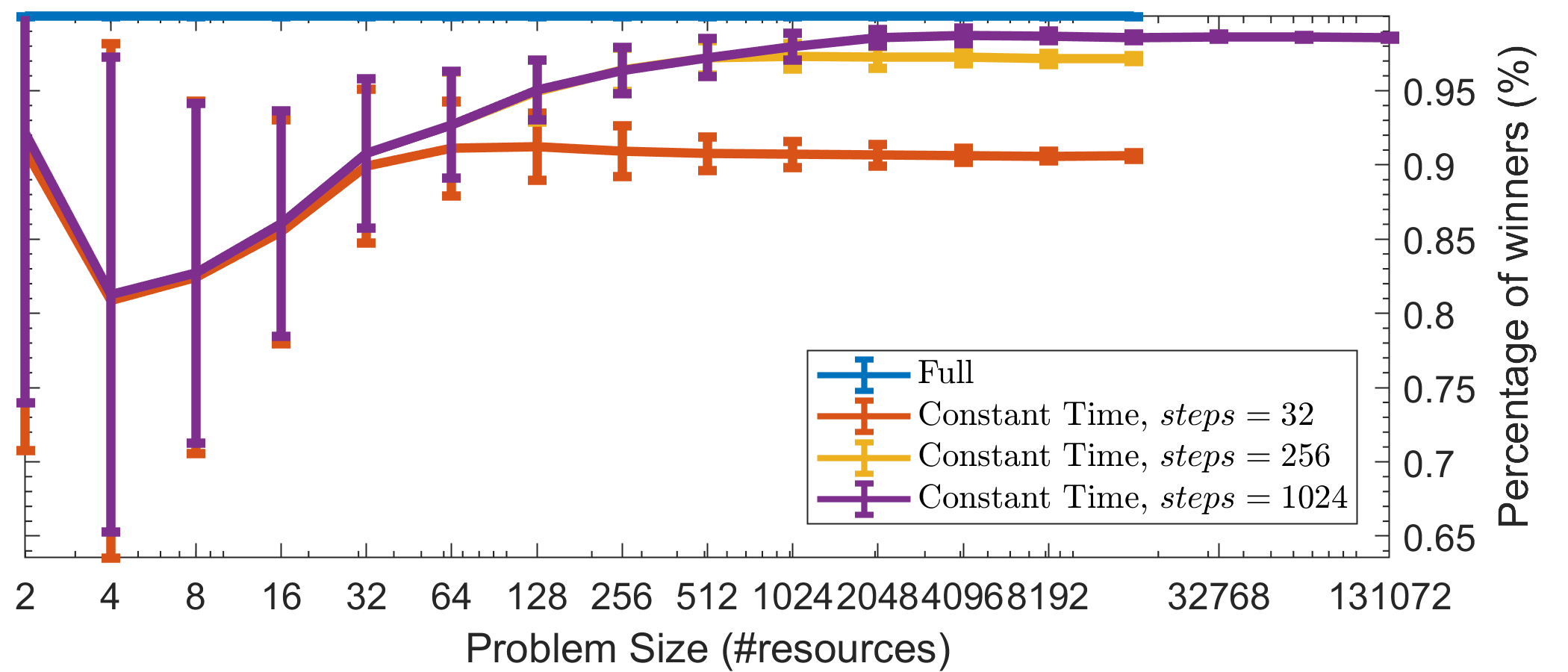}
		\caption{}
		\label{fig: testcase2_gridSpots_4_distance_025_percentageOfWinners}
	\end{subfigure}%
	~
	\begin{subfigure}[t]{0.5\textwidth}
		\centering
		\includegraphics[width = 1 \linewidth]{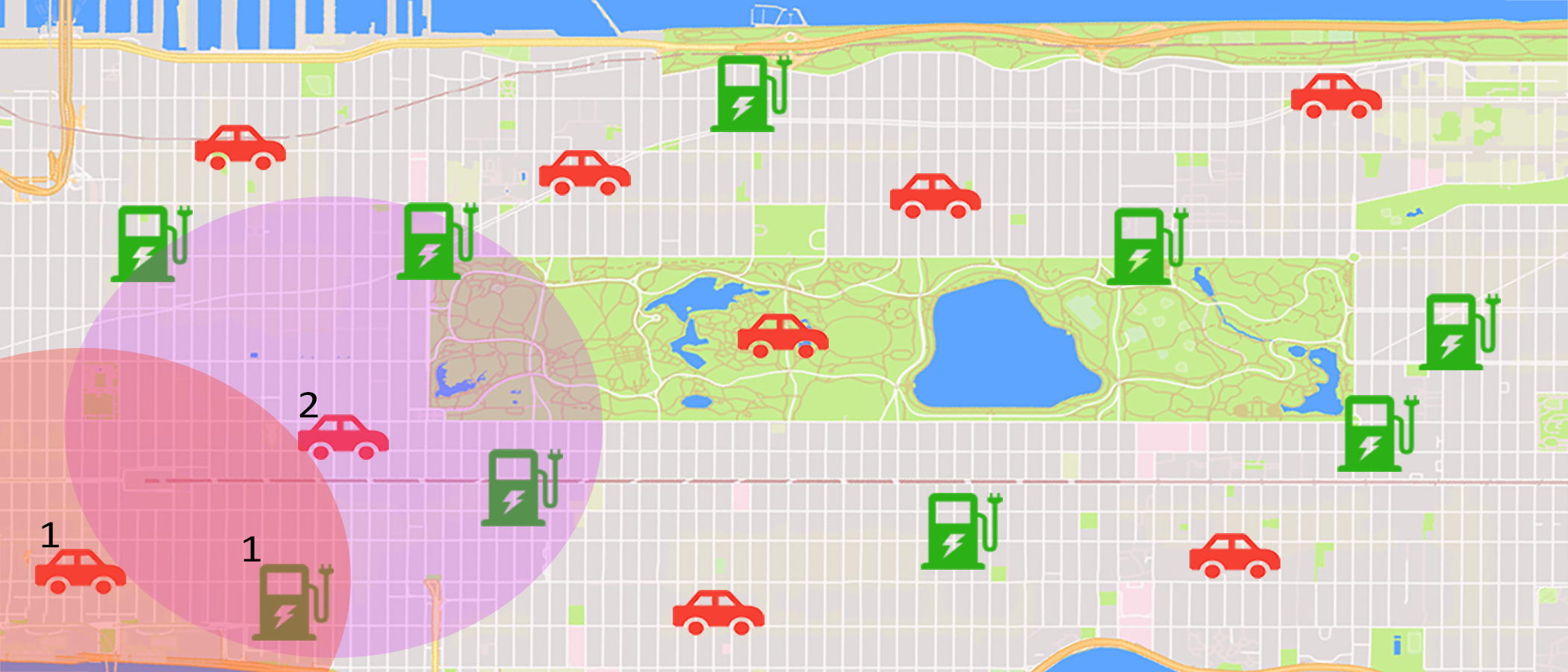}
		\caption{}
		\label{fig: map}
	\end{subfigure}
	\caption{Left to right, top to bottom: (\ref{fig: testcase2_mapCutoffBounded_gridSpots_4_mapCutoffDistance_99999999_stepsPerAgent}) Average time (\#steps) for an agent to acquire a resource, (\ref{fig: testcase2_mapCutoffBounded_gridSpots_4_mapCutoffDistance_99999999_steps}) Total convergence time (\#steps), (\ref{fig: testcase2_mapCutoffBounded_gridSpots_4_mapCutoffDistance_99999999_computationTime}) Computation time (ns), (\ref{fig: testcase2_mapCutoffBounded_gridSpots_4_mapCutoffDistance_99999999_cumulativeRegret}) Relative difference in SW (\%), (\ref{fig: testcase2_mapCutoffBounded_gridSpots_4_mapCutoffDistance_99999999_percentageOfWinners}) Percentage of `winners', (\ref{fig: testcase2_gridSpots_4_distance_025_cumulativeRegret}) Relative difference in SW (\%) in interrupted execution, (\ref{fig: testcase2_gridSpots_4_distance_025_percentageOfWinners}) Percentage of `winners' in interrupted execution. The aforementioned are for increasing number of resources, and $N = R$. Fig. (\ref{fig: testcase2_mapCutoffBounded_gridSpots_4_mapCutoffDistance_99999999_stepsPerAgent}), (\ref{fig: testcase2_mapCutoffBounded_gridSpots_4_mapCutoffDistance_99999999_steps}), and (\ref{fig: testcase2_mapCutoffBounded_gridSpots_4_mapCutoffDistance_99999999_computationTime}) are in double log scale, while the rest are in single log scale. Fig. (\ref{fig: map}) presents an example of the studied resource allocation scenario in an urban environment. We assume grid length of $\sqrt{4 \times N}$.}
	\label{fig: testcase 2}
\end{figure*}

\subsection{Test Case \#2: Resource Allocation in a Cartesian Map with Manhattan Distances} \label{testcase 2}

Adopting a simple rule allows the applicability of ALMA to large scale multi-agent systems. In this section we will analyze such a scenario. Specifically we are interested in resource allocation in urban environments (e.g., parking spots / charging stations for autonomous vehicles, taxi - passenger matchings, etc.). The aforementioned problems become ever more relevant due to rapid urbanization, and the natural lack of coordination in the usage of resources \cite{varakantham2016sequential}. The latter result in the degradation of response (e.g., waiting time) and quality metrics in large cities \cite{varakantham2016sequential}.

\subsubsection{Setting}

Let us consider a Cartesian map representing a city on which are randomly distributed vehicles and charging stations, as depicted in Fig. \ref{fig: map}. The utility received by a vehicle $n$ for using a charging station $r$ is proportional to the inverse of their distance, i.e., $u_n(r) = 1 / d_{n,r}$. Since we are in an urban environment, let $d_{n,r}$ denote the Manhattan distance. Typically, there is a cost each agent is willing to pay to drive to a resource, thus there is a cut-off distance, upon which the utility of acquiring the resource is zero (or possibly negative). This is a typical scenario encountered in resource allocation in urban environments, where there are spatial constraints and local interactions.

The way such problems are typically tackled, is by dividing the map to sub-regions, and solving each individual sub-problem. For example, Singapore is divided into 83 zones based on postal codes \cite{6040748}, and taxi drivers' policies are optimized according to those \cite{nguyen2017collective,varakantham2012decision}. On the other hand, not placing bounds means that the current solutions will not scale. To the best of our knowledge, we are the first to propose an anytime heuristic for resource allocation in urban environments that can scale in constant time without the need to artificially split the problem. Instead, ALMA exploits the two typical characteristics of an urban environment: the anonymity in interactions and homogeneity in supply and demand \cite{varakantham2016sequential} (e.g., assigning any of two equidistant charging stations to a vehicle, would typically result to the same utility). This results in a simple learning rule which, as we will demonstrate in this section, can scale to hundreds of thousands of agents.

\subsubsection{Convergence Time}

To demonstrate the latter, we placed a bound on the maximum number of resources each agent is interested in, and on the maximum number of agents competing for a resource, specifically $R^n = N^r \in \{8, 16, 32, 64, 128\}$. According to Corollary \ref{Th: convergence agent}, bounding these two quantities should result in convergence in constant time, regardless of the total problem size (R, N). The latter is corroborated by Fig. \ref{fig: testcase2_mapCutoffBounded_gridSpots_4_mapCutoffDistance_99999999_stepsPerAgent}, which shows that the average number of time-steps until an agent successfully claims a resource remains constant as we increase the total problem size. Same is true for the system's convergence time (Fig. \ref{fig: testcase2_mapCutoffBounded_gridSpots_4_mapCutoffDistance_99999999_steps}), which caps as $R$ increases. The small increase is due to outliers, as Fig. \ref{fig: testcase2_mapCutoffBounded_gridSpots_4_mapCutoffDistance_99999999_steps} reports the convergence time of the last agent. This results to approximately $7$ orders of magnitude less computation time than the centralized Hungarian algorithm (Fig. \ref{fig: testcase2_mapCutoffBounded_gridSpots_4_mapCutoffDistance_99999999_computationTime}), and this number would grow boundlessly as we increase the total problem size. Moreover, as mentioned, in an actual implementation, any algorithm for the assignment problem would face additional overhead due to communication time, reliability protocols, etc.

\subsubsection{Efficiency}

Along with the constant convergence time, ALMA is able to reach high quality matchings, achieving less than $7.5\%$ worse social welfare (SW) than the optimal (Fig. \ref{fig: testcase2_mapCutoffBounded_gridSpots_4_mapCutoffDistance_99999999_cumulativeRegret}). The latter refers to the small bound of $R^n = 8$. As observed in Section \ref{testcase 1}, with a small number of choices, a single wrong matching can have a significant impact to the final social welfare. By increasing the bound to a more realistic number (e.g., $R^n = 32$), we achieve less than $2.5\%$ worse SW. In general, for $R > 2$ resources and different values of $R^n$, ALMA achieves between $1.9 - 12 \%$ loss in SW, while the greedy approach achieves $8.0 - 24 \%$ and the random $7.3 - 43.4 \%$. The behavior of the graphs depicted in Fig. \ref{fig: testcase2_mapCutoffBounded_gridSpots_4_mapCutoffDistance_99999999_cumulativeRegret} for $R^n \in \{8, 16, 32, 64, 128\}$ indicate that, as the problem size ($R$) increases, the social welfare reaches its lowest value at $R = 2 \times R^n$. To investigate the latter, we have included a graph for increasing $R^n$ (instead of constant to the problem size), specifically $R^n = R / 2$. ALMA achieves a constant loss in social welfare (approximately $11\%$). The greedy approach achieves loss of $14\%$, while the random solution degrades towards $44\%$ loss.

Compared to Test Case \#1, this is a significantly harder problem for a decentralized algorithm with no communication and no global knowledge of the resources. The set of resources each agent is interested in is a proper subset of the set of the total resources, i.e., $\mathcal{R}^n \subsetneq \mathcal{R}$ (or could be $R < N$). Furthermore, the lack of communication between the participants, and the stochastic nature of the algorithm can lead to deadlocks, e.g., in Fig. \ref{fig: map}, if vehicle 2 acquires resource 1, then vehicle 1 does not have an available resource in range. Nonetheless, ALMA results in an almost complete matching. Fig. \ref{fig: testcase2_mapCutoffBounded_gridSpots_4_mapCutoffDistance_99999999_percentageOfWinners}, depicts the percentage of `winners' (i.e., agents that have successfully claimed a resource $r$ such that $u_n(r) > 0$). The aforementioned percentage refers to the total population ($N$) and not the maximum possible matchings (potentially $< N$). As depicted, the percentage of `winners' is more than $90\%$, reaching up to $97.8\%$ for $R^n = 128$. We also employed ALMA in larger simulations with up to $131072$ agents, and equal resources. As seen in Fig. \ref{fig: testcase2_mapCutoffBounded_gridSpots_4_mapCutoffDistance_99999999_percentageOfWinners}, the percentage of winners remains stable at around $98\%$. Even though the size of the problem prohibited us from running the Hungarian algorithm (or an out-of-the-box LP solver) and validating the quality of the achieved matching, the fact that the percentage of winners remains the same suggests that the relative difference in SW will continue on the same trend as in Fig. \ref{fig: testcase2_mapCutoffBounded_gridSpots_4_mapCutoffDistance_99999999_cumulativeRegret}. Moreover, the average steps per agent to claim a resource remains, as proven, constant (Fig. \ref{fig: testcase2_mapCutoffBounded_gridSpots_4_mapCutoffDistance_99999999_stepsPerAgent}). The latter validate the applicability of ALMA in large scale applications with hundreds of thousands of agents. 

\subsubsection{Anytime Property}

In the real world, agents are required to run in real time, which imposes time constraints. ALMA can be used as an anytime heuristic as well. To demonstrate the latter, we compare four configurations: the `full' one, which is allowed to run until the systems converges, and three `constant time' versions which are given a time budget of 32, 256, and 1024 time-steps. In this scenario, we do not impose a bound on $R^n, N^r$, but we assume a cut-off distance, upon which the utility is zero. The cut-off distance was set to $0.25$ of the maximum possible distance, i.e., as the problem size grows, so do the $R^n, N^r$. On par with Test Case \#1, the full version converges in linear time. As depicted in Fig. \ref{fig: testcase2_gridSpots_4_distance_025_cumulativeRegret}, the achieved SW is less than $9\%$ worse than the optimal. The inferior results in terms of SW compared to Test Case \#1 are because this is a significantly harder problem due to the aforementioned deadlocks. On the other hand though, ALMA benefits from the spatial constraints of the problem. The average number of time-steps an individual agent needs to successfully claim a resource is significantly smaller, which suggest that we can enforce computation time constraints. Restricting to only 32, 256, and 1024 time-steps, results in $1.25\%$, $0.12\%$, and $0.03\%$ worse SW than the unrestricted version, respectively. Even in larger simulations with up to $131072$ agents, the percentage of winners (Fig. \ref{fig: testcase2_gridSpots_4_distance_025_percentageOfWinners}) remains stable at $98.6\%$, which suggests that the relative difference in SW will continue on the same trend as in Fig. \ref{fig: testcase2_gridSpots_4_distance_025_cumulativeRegret} (we do not suggest that this is the case in any domain. For example, in the noisy common preferences domain of Test Case \#1, the quality of the achieved matching decreases boundlessly as we decrease the alloted time. Nevertheless, the aforedescribed domain is a realistic one, with a variety of real-world applications). Finally, the repeated nature of such problems suggests that even in the case of a deadlock, the agent which failed to win a resource, will do so in some subsequent round.

\subsection{Test Case \#3: On-line Taxi Request Match} \label{testcase 3}

In this section we present a motivating test case involving ride-sharing, via on-line taxi request matching, using \emph{real} data of taxi rides in New York City. Ride-sharing (or carpooling), offers great potential in congestion relief and environmental sustainability in urban environments. In the past few years, several commercially successful ride-sharing companies have surfaced (e.g., Uber, Lyft, etc.), giving rise to a new incarnation of ride-sharing: dynamic ride-sharing, where passengers are matched in real-time. Ride-sharing, though, results to some passenger disruption due to loss in flexibility, security concerns, etc. Compensation comes in the form of monetary incentives, as it allows passengers to share the travel expenses, and thus reduce the cost. Ride-sharing companies account for a plethora of factors, like current demand, prime time pricing, the cost of the route without any ride-sharing, the likelihood of a match, etc. Yet, a fundamental factor of the cost of any shared ride, no matter if it is a company or a locally-organized car sharing scheme, is the traveled distance.

In this test case, we attempt to maximize the total distanced saved, by matching taxi requests of high overlap. Fig. \ref{fig: testcase3_taxis_map} provides an illustrative example. There are two passengers (depicted as yellow and red) with high overlap routes. Each can drive on their own to their respective destinations (dashed yellow and red line respectively), or share a ride (green line) and reduce travel costs.

Dynamic ride-sharing is an inherently \emph{on-line} setting, as a matching algorithm is unaware of the requests that will appear in the future and needs to make decisions for the requests before they `expire' (a similar setting was studied in \cite{ashlagi2018maximum}). ALMA is highly befitting for such a scenario, as it involves large-scale matchings under dynamic demand, it is highly decentralized, and partially observable.

\subsubsection{Setting}

We use a dataset \footnote{\url{https://www.kaggle.com/debanjanpaul/new-york-city-taxi-trip-distance-matrix/}} of all taxi requests $(\rho)$ in New York City during one week (\{01-01-16 0:00 - 01-07-16 23:59\}, 34077 requests in total). The data include pickup and drop-off times, and geolocations. Requests appear (become open) at their respective pickup time, and wait $k_\rho$ time-steps to find a match. Let a time-step be one minute. After $k_\rho$ time-steps we call request $\rho$, critical. If a critical request is not matched, we assume they drive off to their destination in a single passenger ride. Let $open, critical$ denote the sets of open, and critical requests respectively, and let $current = open \cup critical$. To compute $k_\rho$ we assume the following: There is a minimum $minW$, and a maximum $maxW$ waiting time set by the ride-sharing company, i.e., $minW \leq k_\rho \leq maxW, \forall \rho$. Moreover, since each passenger specifies his destination in advance, we can compute the trip time ($l_\rho$). Assuming people are willing to wait for a time that is proportional to their trip time, let $k_\rho = q \times l_\rho$, where $q \in [0, 1]$. The parameters $minW, maxW$, and $q$ can be set by the ride-sharing company. We report results on different values for all of the above parameters. For each pair $\rho_1, \rho_2$ of requests, we compute the driving distance ($d_{\rho_1, \rho_2} = \min$ of all possible combinations of driving between $\rho_1, \rho_2$'s pickup and drop-off locations) that would be traveled if $\rho_1, \rho_2$ are matched, i.e., if they share the same taxi. Subsequently, the utility of matching $\rho_1$ to $\rho_2$ (distance saved) is $u_{\rho_1}(\rho_2) = d_{\rho_1, \rho_2}$ (km). 

Given the on-line nature of the setting, it might be beneficial to use the following \emph{non-myopic heuristic}: avoid matching low utility pairs, as long as the requests are not critical, since more valuable pairs might be presented in the future. Thus, if $u_{\rho_1}(\rho_2) < d_{min}$, and $\rho_1, \rho_2 \notin critical$, we do not match $\rho_1$, $\rho_2$. In what follows, we select for each algorithm and for each simulation the value $d_{min} \in \{0, 500, 1000, 1500, 2000, 2500\}$ that results in the highest score. To compute the actual trip time, and driving distance, we have used the Open Source Routing Machine \footnote{\url{http://project-osrm.org/}}, which computes shortest paths in road networks.

\begin{figure}[t!]
	\centering
	\includegraphics[width = 1 \linewidth]{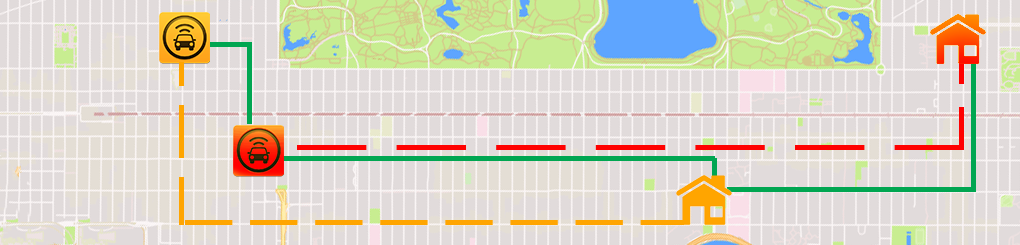}
	\caption{Example of the studied taxi request matching scenario.}
	\label{fig: testcase3_taxis_map}
\end{figure}

\subsubsection{Computation of the optimal matching}

In this scenario, each request has a dual role, being both an agent and a resource, i.e., we have a non-bipartite graph. For computing the optimal (maximum weight) matching for the employed on-line heuristics, we use the blossom algorithm of \cite{edmonds1965maximum}, which computes a maximum weight matching in general graphs. This enables us to compute the best possible matching among current requests, i.e., requests that have not `expired' at the time of the computation. In fact, the following observation allows us to compute the optimal \emph{off-line} matching as well, i.e., the best possible matching over the whole time interval. Let $a_\rho$ be the pick-up time of request $\rho$ and let $e_\rho = a_\rho + k_\rho$ be the time when it becomes critical. Then, we can redefine the utility of a matching as:

\begin{equation} \label{eq: testcase3_optimaloffline}
	u_{\rho_1}(\rho_2) =
	\begin{cases}
		d_{\rho_1, \rho_2}, &\text{if } a_i \leq e_j \text{ and } e_i \leq e_j\\
		-1,  &\text{otherwise}
	\end{cases}
\end{equation}

\noindent
i.e., the utility is the distance saved if both requests are simultaneously \emph{active} when matched and $-1$ otherwise. The latter effectively results in this pair never being matched in an optimal solution. We remark that this clairvoyant matching is not feasible in the on-line setting and serves as a benchmark against which we can compare the performance of on-line algorithms, as is common in the literature of competitive analysis \cite{borodin2005online}. The measure of efficiency, as compared to the off-line optimal, will be the \emph{empirical competitive ratio}, i.e., the ratio of the social welfare of the on-line algorithm over the welfare of the optimal, as measured empirically for our dataset.

\subsubsection{The Blossom Algorithm vs Linear Programming} 

Contrary to the other test cases, the fact that the graph is not bipartite has implications on how linear programming can be used to compute the optimal solution. In particular, in all of the presented test cases (\#1, \#2, and \#3), one can formulate the problem of computing the optimal solution as an Integer Linear Program (ILP) and then solve it using some general solver like CPLEX \footnote{https://www.ibm.com/analytics/cplex-optimizer}. Yet, the computational complexity of solving the aforementioned ILP varies amongst the different test cases.

Solving integer linear programs is generally quite computationally demanding. Thus, it is common to resort to solving the \emph{LP relaxation} (where the integrality constraints have been `relaxed' to fractional constraints), which can be computed in polynomial time.

In the case of bipartite graphs (such as test cases \#1 and \#2), the standard LP relaxation admits integer solutions, i.e., solutions to the actual maximum weight bipartite matching problem (because the constraint matrix is totally unimodular). In the case of general (non-bipartite) graphs, this is no longer the case, thus we have to resort to solving the \emph{LP relaxation}. It is known that the (fractional) optimal solution to the LP relaxation might be better than the (integral) optimal solution to the original ILP formulation (i.e., it has an \emph{integrality gap} which is larger than $1$, in fact $2$). In other words, solving the LP relaxation will not provide solutions for the maximum weight matching problem but for an `easier version' with fractional matchings. Thus, comparing against that solution could only give very pessimistic ratios, when the real empirical competitive ratios are much better. 

One can derive a different ILP formulation of the problem using more involved constraints (called `blossom' constraints \cite{edmonds1965maximum}), whose relaxation admits integer solutions (i.e., the integrality gap is now $1$). However, the latter would result in an exponential number of constraints, and one would need to employ the Ellipsoid method with an appropriately chosen separation oracle to solve it in polynomial time (see \cite[page 4]{feige2002approximating} for more details]. Overall, the employment of the combinatorial algorithm of \cite{edmonds1965maximum} for finding the maximum weight matching is a cleaner and more efficient solution.

\subsubsection{Benchmarks}

Each request runs ALMA independently. ALMA waits until the request becomes critical, and then matches it by running Alg. \ref{algo: learning rule}, where $\mathcal{N} = critical$, and $\mathcal{R} = current$. In this non-bipartite scenario, if an agent is matched under his dual role as a resource, he is immediately removed. As we explained earlier, it is infeasible for an on-line algorithm to compute the optimal matching over the whole period of time. Instead, we consider \emph{just-in-time} and \emph{in batches} optimal solutions. Specifically, we compare to the following \cite{agatz2011dynamic,ashlagi2018maximum}:

\begin{itemize}
	\item \textbf{Just-in-time Max Weight Matching (JiTMWM)}: Waits until a request becomes critical and then computes a maximum-weight matching of all the current requests, i.e $\mathcal{N} = \mathcal{R} = current$.
	\item \textbf{Batching Max Weight Matching (BMWM)}: Waits $x$ time-steps and then computes a maximum-weight matching of all the current requests, i.e $\mathcal{N} = \mathcal{R} = current$.
	\item \textbf{Batching Greedy (BG)}: Waits $x$ time-steps and then greedily matches current requests, i.e $\mathcal{N} = \mathcal{R} = current$ (ties are broken randomly). Unmatched open requests are removed. For batch size $x = 1$ we get the simple greedy approach where requests are matched as soon as they appear.
\end{itemize}

\begin{figure}[t!]
	\centering
	\begin{subfigure}[t]{0.5\textwidth}
		\centering
		\includegraphics[width = 1 \linewidth]{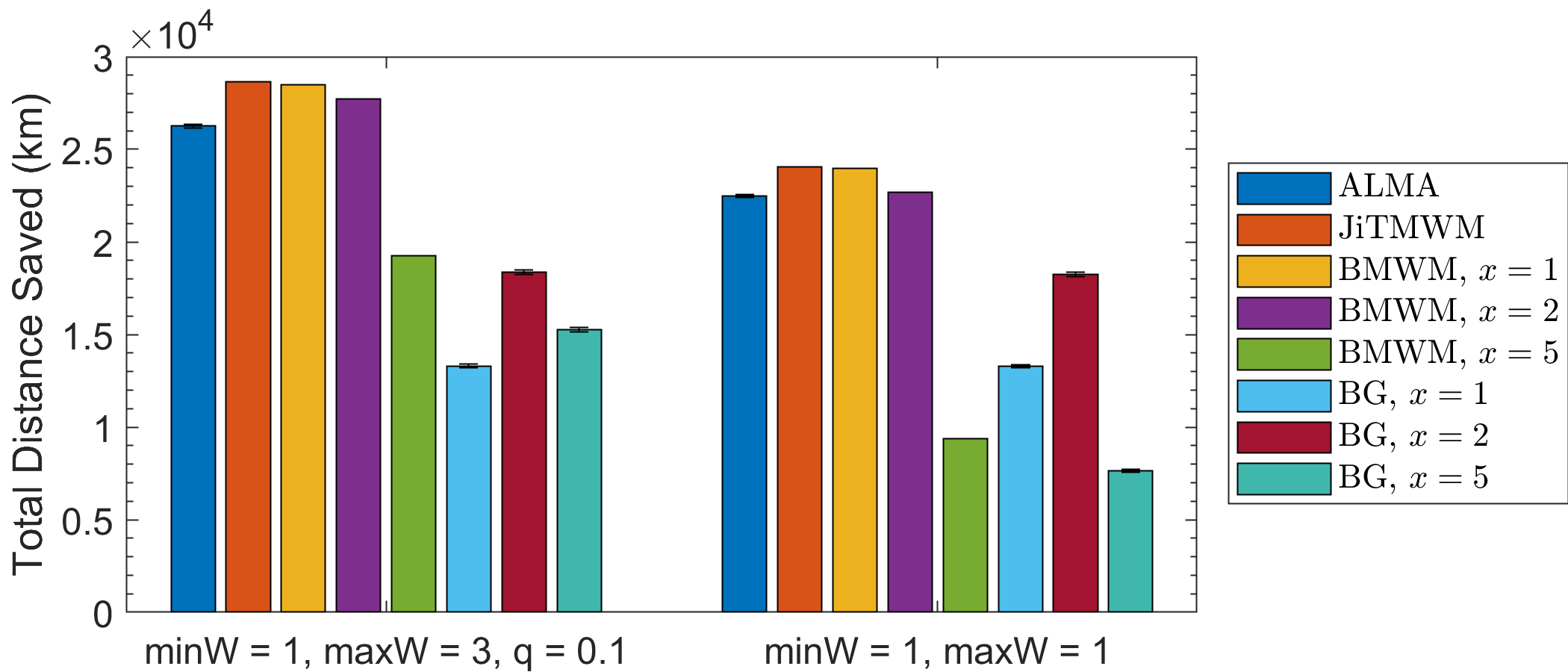}
		\caption{}
		\label{fig: testcase3_bar}
	\end{subfigure}
	~ \\ 
	\begin{subfigure}[t]{0.5\textwidth}
		\centering
		\includegraphics[width = 1 \linewidth]{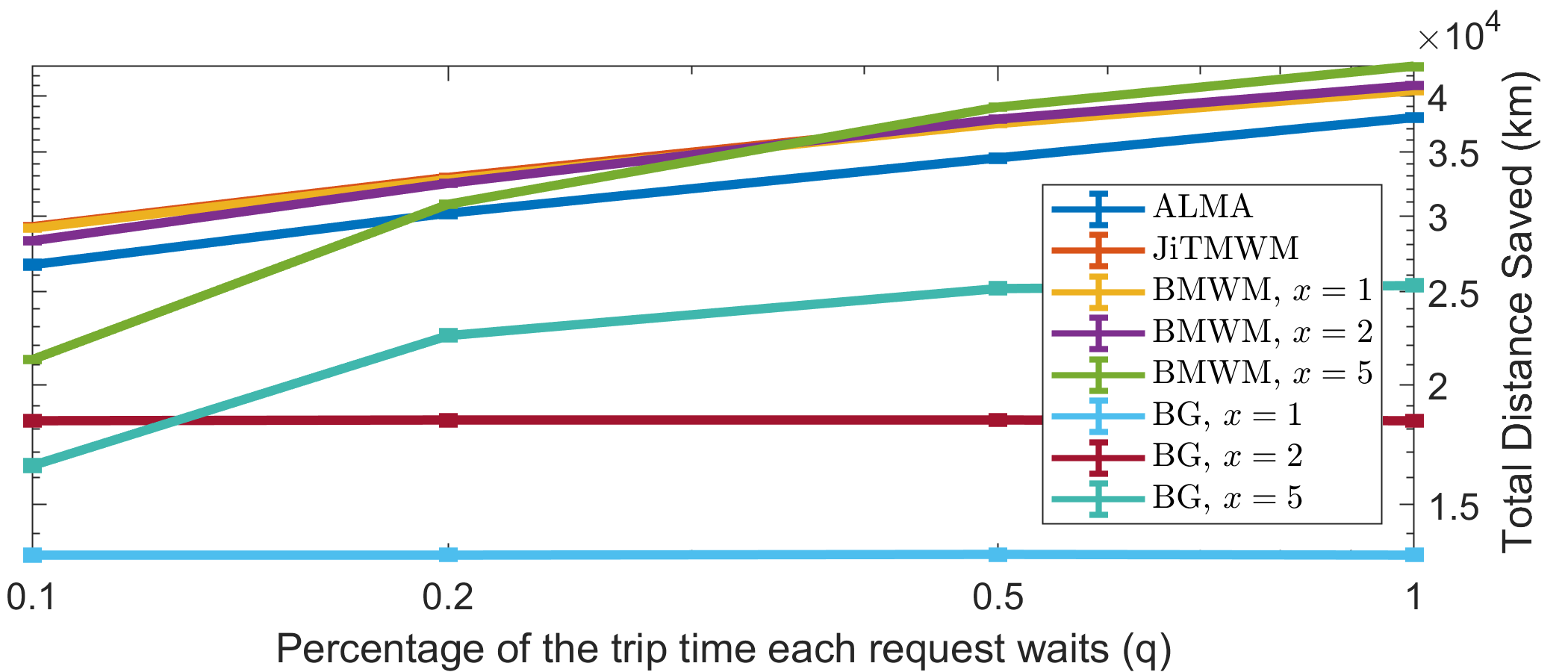}
		\caption{}
		\label{fig: testcase3_graph}
	\end{subfigure}%
	\caption{Total distance saved (km) for various values of $minW, maxW, q$. (\ref{fig: testcase3_bar}, left) Pragmatic scenario, (\ref{fig: testcase3_bar}, right) Requests become critical in just one time-step, (\ref{fig: testcase3_graph}) Various levels of waiting time ($q$) assuming no bounds, i.e., $minW = 0, maxW = \infty$, (double log. scale).}
	\label{fig: testcase 3}
\end{figure}

\begin{table*}[t]
	\centering
	\caption{Empirical Competitive Ratio.}
	\begin{tabular}{@{}lcccccccc@{}}
	\toprule
	$(minW, maxW, q)$  & ALMA   & JiTMWM & BMWM, $x = 1$ & BMWM, $x = 2$ & BMWM, $x = 5$ & BG, $x = 1$ & BG, $x = 2$ & BG, $x = 5$ \\ \midrule
	$(1, 3, 0.1)$      & 0.7890 & 0.8621 & 0.8568        & 0.8283        & 0.5883        & 0.3991      & 0.5422      & 0.4603      \\
	$(1, 1, -)$        & 0.8491 & 0.9243 & 0.9190        & 0.8663        & 0.3810        & 0.5112      & 0.6891      & 0.3038      \\
	$(0, \infty, 0.1)$ & 0.7835 & 0.8528 & 0.8486        & 0.8211        & 0.6221        & 0.3900      & 0.5299      & 0.4840      \\
	$(0, \infty, 0.2)$ & 0.7546 & 0.8207 & 0.8158        & 0.8000        & 0.7616        & 0.3399      & 0.4604      & 0.5688      \\
	$(0, \infty, 0.5)$ & 0.6939 & 0.7439 & 0.7368        & 0.7448        & 0.7668        & 0.2731      & 0.3714      & 0.5133      \\
	$(0, \infty, 1.0)$ & 0.6695 & 0.7390 & 0.7254        & 0.7343        & 0.7706        & 0.2440      & 0.3306      & 0.4606      \\
	\bottomrule
	\end{tabular}
	\label{table: empirical competitive ratio}
\end{table*}

\subsubsection{Efficiency}

Starting with the social welfare, Fig. \ref{fig: testcase 3} presents the total distance saved (km) for various values of $minW, maxW$, and $q$. ALMA loses $8.3\%$ of SW in the pragmatic scenario of Fig. \ref{fig: testcase3_bar} (left), and $6.5\%$ when the requests become critical in just one time-step (Fig. \ref{fig: testcase3_bar} (right)). If no bounds are placed on the minimum and maximum waiting time (i.e., $minW = 0, maxW = \infty$), ALMA exhibits loss of $8 - 11.5 \%$, for various values of $q$ (Fig. \ref{fig: testcase3_graph}). The above are compared to the best performing benchmark on each scenario (JiTMWM, or BMWM, $x = 5$). Moreover, it significantly outperforms every greedy approach. In the first scenario the BGs lose between $35.8 - 53.5 \%$, in the second between $24 - 68.2 \%$, and in the third between $31.5 - 69 \%$. 

Once more, it is worth noting that ALMA requires just a broadcast of a single bit to indicate the occupancy of a resource, while the compared approaches require either \emph{message exchange} for sharing the utility table, or the use of a \emph{centralized authority}. For example, the greedy solution would require message exchange to communicate users' preferences and resolve collisions in a decentralized setting, and every batching approach would require a common centralized synchronization clock.

Table \ref{table: empirical competitive ratio} presents the empirical competitive ratio for the first day of the week of the employed dataset. As we can see, even in the extreme, unlikely scenarios where we assume that people would be willing to wait for more than 10 or 20 minutes (which correspond to large values of $q$), ALMA achieves high relative efficiency, compared to the off-line (infeasible) benchmark. These scenarios are favorable to the off-line optimal, because requests stay longer in the system and therefore the algorithm takes more advantage of its foreseeing capabilities (hence, the drop in the competitive ratios is observed in \emph{all} of the employed algorithms). In particular, ALMA achieves an empirical competitive ratio of $0.67$ for $q = 1$ and even better ratios for more realistic scenarios (as large as $0.85$). The just-in-time and batch versions of the maximum weight matching perform slightly better, but this is to be expected, as they compute the maximum weight matching on the graphs of the current requests. In spite of the unpredictability of the on-line setting, and the dynamic nature of the demand, ALMA is consistently able to exhibit high performance, in all of the employed scenarios.

\section{Conclusion} \label{Conclusion}

Algorithms for solving the assignment problem, whether centralized or distributed, have runtime that increases with the total problem size, even if agents are interested in a small number of resources. Thus, they can only handle problems of some bounded size. Moreover, they require a significant amount of inter-agent communication. Humans on the other hand are routinely called upon to coordinate in large scale in their everyday lives, and are able to fast and robustly match with resources under dynamic and unpredictable demand. One driving factor that facilitates human cooperation is the principle of altruism. Inspired by human behavior, we have introduced a novel anytime heuristic (ALMA) for weighted matching in both bipartite and non-bipartite graphs. ALMA is decentralized and requires agents to only receive partial feedback of success or failure in acquiring a resource. Furthermore, the running time of the heuristic is constant in the total problem size, under reasonable assumptions on the preference domain of the agents. As autonomous agents proliferate (e.g., IoT devices, intelligent infrastructure, autonomous vehicles, etc.), having robust algorithms that can scale to hundreds of thousands of agents is of utmost importance.

The presented results provide an empirical proof of the high quality of the achieved solution in a variety of scenarios, including both synthetic and \emph{real} data, time constraints and on-line settings. Furthermore, both the proven theoretical bound (which guarantees constant convergence time), and the computation speed comparison (which grounds ALMA to a proven fast centralized algorithm), argue for its applicability to large scale, real world problems. As future work, it would be interesting to identify meaningful domains in which ALMA can provide provable worst-case performance guarantees, as well as to empirically evaluate its performance on other real datasets, corresponding to important real-world, large-scale problems.

\bibliographystyle{alpha}
\bibliography{bibliography}

\clearpage

\appendix
\section{Appendix}

\subsection{Proof of Theorem \ref{Th: convergence system}} \label{Convergence system proof}

\begin{customthm}{\ref{Th: convergence system}}
	For $N$ agents and $R$ resources, the expected number of steps until the system of agents following Alg. \ref{algo: learning rule} converges to a complete matching is bounded by (\ref{Eq: convergence bound system appendix}), where $p^* = f(loss^*)$ and $loss^*$ is given by Eq. \ref{Eq: loss* system appendix}.

	\begin{equation} \label{Eq: convergence bound system appendix}
		\mathcal{O}\left( R \frac{2 - p^*}{2 (1 - p^*)} \left(\frac{1}{p^*} \log N + R \right) \right) 
	\end{equation}\normalsize

	\begin{equation} \label{Eq: loss* system appendix}
		loss^* = \underset{loss_n^r}{\argmin} \left( \underset{r \in \mathcal{R}, n \in \mathcal{N}}{\min}(loss_n^r), 1 - \underset{r \in \mathcal{R}, n \in \mathcal{N}}{\max}(loss_n^r) \right)
	\end{equation}\normalsize
\end{customthm}

In this section we provide a formal proof of Theorem \ref{Th: convergence system}. \footnote{The proof is an adaptation of the convergence proof of \cite{cigler2011reaching} and \cite{DANASSIS:2019}.} To facilitate the proof, we will initially assume that every agent, on every collision, backs-off with the same constant probability, i.e.,:

\begin{equation} \label{Eq: constant p}
	P_n(r, \prec_n) = p > 0, \forall n \in \mathcal{N}, \forall r \in \mathcal{R}
\end{equation}\normalsize

\subsubsection{Case \#1: Multiple Agents, Single resource ($R = 1$)} \label{Case 1}

We will describe the execution of the proposed learning rule as a discrete time Markov chain (DTMC) \footnote{For an introduction on Markov chains see \cite{norris1998markov}}. In every time-step, each agent performs a Bernoulli trial with probability of `success' $1 - p$ (remain in the competition), and failure $p$ (back-off). When $N$ agents compete for a single resource, a state of the system is a vector $\{0, 1\}^N$ denoting the individual agents that still compete for that resource. But, since the back-off probability is the same for everyone (Eq. \ref{Eq: constant p}), we are only interested in how many agents are competing and not which ones. Thus, in the single resource case ($R = 1$), we can describe the execution of the proposed algorithm using the following chain:

\begin{definition} \label{def: markov chain X}
	Let $\{X_t\}_{t \geq 0}$ be a DTMC on state space $S = \{0, 1, \ldots ,N\}$ denoting the number of agents still competing for the resource. The transition probabilities are as follows:

	\begin{align*}
		& Pr(X_{t + 1} = N | X_t = 0) = 1 & \text{restart} \\
		& Pr(X_{t + 1} = 1 | X_t = 1) = 1 & \text{absorbing} \\
		& Pr(X_{t + 1} = j | X_t = i) = \binom{i}{j} p^{i - j} (1 - p)^j & i > 1, j \leq i \\
	\end{align*}\normalsize
	\noindent
	(all the other transition probabilities are zero)
\end{definition}

Intuitively, this Markov chain describes the number of individuals in a decreasing population, but with two caveats: The goal (absorbing state) is to reach a point where only one individual remains, and if we reach zero, we restart.

Before proceeding with Theorem \ref{Th: convergence system}'s convergence proof, we will restate Mityushin's Theorem \cite{rego1992naive} for hitting time bounds in Markov chains, define two auxiliary DTMCs, and prove two auxiliary lemmas.

\begin{theorem} \label{Th: Mityushin} (Mityushin's Theorem \cite{rego1992naive})
	Let $A = \{0\}$ be the absorbing state of a Markov chain $\{X_t\}_{t \geq 0}$. If $\mathds{E}(X_{t + 1} | X_t = i) < \frac{i}{\beta}$, $\forall i \geq 1$ and some $\beta > 1$, then:

	\begin{equation} \label{Eq: Mityushin}
		\mathds{E}(T_i^A) < \lceil \log_{\beta} i \rceil + \frac{\beta}{\beta - 1}
	\end{equation}\normalsize
	where $T_i^A$ denotes the hitting time of a state in $A$, starting from state $i$.
\end{theorem}

\begin{definition} \label{def: markov chain Y}
	Let $\{Y_t\}_{t \geq 0}$ be a DTMC on state space $S = \{0, 1, \ldots ,N\}$ with the following transition probabilities (two absorbing states, 0 and 1):

	\begin{align*}
		& Pr(Y_{t + 1} = 0 | Y_t = 0) = 1 & \text{absorbing} \\
		& Pr(Y_{t + 1} = 1 | Y_t = 1) = 1 & \text{absorbing} \\
		& Pr(Y_{t + 1} = j | Y_t = i) = \binom{i}{j} p^{i - j} (1 - p)^j & i > 1, j \leq i \\
	\end{align*}\normalsize
	\noindent
	(all the other transition probabilities are zero)
\end{definition}

\begin{definition} \label{def: markov chain Z}
	Let $\{Z_t\}_{t \geq 0}$ be a DTMC on state space $S = \{0, 1, \ldots ,N\}$ with the following transition probabilities (state 0 the only absorbing state):

	\begin{align*}
		& Pr(Z_{t + 1} = 0 | Z_t = 0) = 1 & \text{absorbing} \\
		& Pr(Z_{t + 1} = j | Z_t = i) = \binom{i}{j} p^{i - j} (1 - p)^j & i \geq 1, j \leq i \\
	\end{align*}\normalsize
	\noindent
	(all the other transition probabilities are zero)
\end{definition}

\begin{lemma} \label{lm: hitting time of Z}
	The expected hitting time of the set of absorbing states $A = \{0\}$, starting from state $Z_0 = N$, of the DTMC $\{Z_t\}$ of Definition \ref{def: markov chain Z} is bounded by $\mathcal{O}\left( \frac{1}{p} \log N \right)$.
\end{lemma}

\begin{proof}
	If the DTMC $\{Z_t\}$ is in state $Z_t = i$, the next state $Z_{t + 1}$ is drawn from a binomial distribution with parameters $(i, 1 - p)$. Thus, the expected next state is $\mathds{E}(Z_{t + 1} | Z_t = i) = i(1 - p)$. Using Theorem \ref{Th: Mityushin} with $\beta = \frac{1}{1 - p}$ results in the required bound:

	\begin{equation}
		\mathds{E}(T_N^A) = \mathcal{O}\left( \frac{1}{p} \log N \right)
	\end{equation}\normalsize
\end{proof}

\begin{corollary} \label{lm: hitting time of Y}
	The expected hitting time of the set of absorbing states $A = \{0, 1\}$, starting from state $Y_0 = N$, of the DTMC $\{Y_t\}$ of Definition \ref{def: markov chain Y} is bounded by $\mathcal{O}\left( \frac{1}{p} \log N \right)$.
\end{corollary}

\begin{proof}
	The expected hitting time of the absorbing state of $\{Z_t\}$ is an upper bound of the expected hitting time of $\{Y_t\}$. This is because any path that leads into state 0 in $\{Z_t\}$ either does not go through state 1 (thus happens with the same probability as in $\{Y_t\}$), or goes through state 1. But, state 1 in $\{Y_t\}$ is an absorbing state, hence in the latter case the expected hitting time for $\{Y_t\}$ would be one step shorter.
\end{proof}

Let $h_i^A$ denote the hitting probability of a set of states $A$, starting from state $i$. We will prove the following lemma.

\begin{lemma} \label{lm: probability of entering state 1 Y}
	The hitting probability of the absorbing state $\{1\}$, starting from any state $i \geq 1$, of the DTMC $\{Y_t\}$ of Definition \ref{def: markov chain Y} is given by Eq. \ref{Eq: hitting probability bound}. This is a tight lower bound.

	\begin{equation} \label{Eq: hitting probability bound}
		h_i^{\{1\}} = \Omega\left( \frac{2(1 - p)}{2 - p} \right), \forall i \geq 1
	\end{equation}\normalsize
\end{lemma}

\begin{proof}
	For simplicity we denote $h_i \overset{\Delta}{=} h_i^{\{1\}}$. We will show that for $p \in (0, 1)$, $h_i \geq \lambda = \frac{2(1 - p)}{2 - p}, \forall i \geq 1$ using induction. First note that since state $\{0\}$ is an absorbing state, $h_0 = 0$, $h_1 = 1 \geq \lambda$ and that $\lambda \in (0, 1)$.

	The vector of hitting probabilities $h^A = (h_i^A: i \in S = \{0, 1, \ldots ,N\})$ for a set of states $A$ is the minimal non-negative solution to the system of linear equations \ref{hitting probabilities system}:

	\begin{equation} \label{hitting probabilities system}
		\begin{cases}
			h_i^A = 1, &\text{ if } i \in A \\
			h_i^A = \underset{j \in S}{\sum} p_{ij}h_j^A, &\text{ if } i \notin A
		\end{cases}
	\end{equation}\normalsize

	By replacing $p_{ij}$ with the probabilities of Definition \ref{def: markov chain Y}, the system of equations \ref{hitting probabilities system} becomes:

	\begin{equation} \label{hitting probabilities}
		\begin{cases}
			h_i^A = 1, &\text{ if } i \in A \\
			h_i^A = \underset{j = 0}{\overset{i}{\sum}} \binom{i}{j} p^{i - j} (1 - p)^j h_j^A, &\text{ if } i \notin A
		\end{cases}
	\end{equation}\normalsize

	\noindent
	Base case:

	\begin{align*}
		h_2 &= (1 - p)^2 h_2 + 2p(1 -p) h_1 + p^2 h_0 = \frac{2p(1 - p)}{1 - (1 - p)^2} \\
			&= \frac{2(1 - p)}{2 - p}\geq \lambda
	\end{align*}\normalsize

	\noindent
	Inductive step: We assume that $\forall j \leq i - 1 \Rightarrow h_j \geq \lambda$. We will prove that $h_i \geq \lambda, \forall i > 2$.

	\begin{align*}
		h_i &= \underset{j = 0}{\overset{i}{\sum}} \binom{i}{j} p^{i - j} (1 - p)^j h_j \\
		&= p^i h_0 + i p^{i - 1} (1 - p) h_1 + \underset{j = 2}{\overset{i - 1}{\sum}} \binom{i}{j} p^{i - j} (1 - p)^j h_j \\
		&+ (1 - p)^i h_i \\
		&\geq p^i h_0 + i p^{i - 1} (1 - p) h_1 + \underset{j = 2}{\overset{i - 1}{\sum}} \binom{i}{j} p^{i - j} (1 - p)^j \lambda \\
		&+ (1 - p)^i h_i \\
		&= i p^{i - 1} (1 - p) + [1 - p^i - (1 - p)^i - i p^{i - 1} (1 - p)] \lambda \\
		&+ (1 - p)^i h_i \\
		\Rightarrow h_i &= \frac{i p^{i - 1} (1 - p) + [1 - p^i - (1 - p)^i - i p^{i - 1} (1 - p)] \lambda}{1 - (1 - p)^i} \\
		\Rightarrow h_i &= \lambda - \frac{p^i}{1 - (1 - p)^i} \lambda + \frac{i p^{i - 1} (1 - p)}{1 - (1 - p)^i} (1 - \lambda)
	\end{align*}\normalsize

	\noindent
	We want to prove that:

	\begin{align*}
		h_i &\geq \lambda \Rightarrow \\
		\frac{i p^{i - 1} (1 - p)}{1 - (1 - p)^i} (1 - \lambda) &\geq \frac{p^i}{1 - (1 - p)^i} \lambda \Rightarrow \\
		i p^{i - 1} (1 - p) &\geq [p^i + i p^{i - 1} (1 - p)] \lambda \Rightarrow \\
		\frac{i p^{i - 1} (1 - p) + p^i - p^i}{p^i + i p^{i - 1} (1 - p)} &\geq \lambda \Rightarrow \\
		1 - \frac{p^i}{p^i + i p^{i - 1} (1 - p)} &\geq \lambda \Rightarrow \\
		1 - \frac{p^i}{p^i + i p^{i - 1} (1 - p)} &\geq \frac{2(1 - p)}{2 - p} \Rightarrow \\
		1 - \frac{p^i}{p^i + i p^{i - 1} (1 - p)} &\geq 1 - \frac{p}{2 - p} \Rightarrow \\
		\frac{p^i}{p^i + i p^{i - 1} (1 - p)} &\leq \frac{p}{2 - p} \Rightarrow \\
		p^i (2 - p) &\leq p [p^i + i p^{i - 1} (1 - p)] \Rightarrow \\
		p^i (2 - p) &\leq p^i [p + i (1 - p)] \Rightarrow \\
		2 - 2p - i + ip &\leq 0 \Rightarrow \\
		2 - i - p(2 - i) &\leq 0 \Rightarrow \\
		(2 - i)(1 - p) &\leq 0 \Rightarrow \\
		2 - i &\leq 0 
	\end{align*}\normalsize

	\noindent
	which holds since $i > 2$.

	The above bound is also tight since $\exists i \in S: h_i = \lambda$, specifically $h_2 = \lambda$.
\end{proof}

Now we can prove the following theorem that bounds the convergence time of the DTCM of Definition \ref{def: markov chain X}, which corresponds to the proposed learning rule for the case of a single available resource ($R = 1$) and constant back-off probability.

\begin{theorem} \label{th: convergence, case R = 1}
	The expected hitting time of the set of absorbing states $A = \{1\}$ of the DTMC $\{X_t\}$ of Definition \ref{def: markov chain X}, starting from any initial state $X_0 \in \{0, 1, \ldots ,N\}$, is bounded by:

	\begin{equation} \label{Eq: complexity for R = 1}
		\mathcal{O}\left( \frac{2 - p}{2 p (1 - p)} \log N \right) 
	\end{equation}\normalsize
\end{theorem}

\begin{proof}
	Using Lemma \ref{lm: probability of entering state 1 Y} we can derive that the DTMC $\{X_t\}$ needs in expectation $\frac{1}{\lambda} = \frac{2 - p}{2 (1 - p)}$ passes until it hits state 1. Each pass requires $\mathcal{O}\left( \frac{1}{p} \log N \right)$ steps (Corollary \ref{lm: hitting time of Y}). Thus, the expected hitting time of state $A = \{1\}$ is $\mathcal{O}\left( \frac{2 - p}{2 p (1 - p)} \log N \right)$.
\end{proof}

\subsubsection{Case \#2: Multiple Agents, Multiple resources ($R > 1$)} \label{Case 2}

\begin{theorem} \label{Th: convergence, case R > 1}
	For $N$ agents and $R$ resources, assuming a constant back-off probability for each agent, i.e., $P_n(r, \prec_n) = p > 0, \forall n \in \mathcal{N}, \forall r \in \mathcal{R}$, the expected number of steps until the system of agents following of Alg. \ref{algo: learning rule} converges to a complete matching is bounded by (\ref{Eq: convergence time, constant p}).

	\begin{equation} \label{Eq: convergence time, constant p}
		\mathcal{O}\left( R \frac{2 - p}{2 (1 - p)} \left(\frac{1}{p} \log N + R \right) \right) 
	\end{equation}\normalsize
\end{theorem}

\begin{proof}
	At most $N$ agents can compete for each resource. We call this period a round. During a round, the number of agents competing for a specific resource monotonically decreases, since that resource is perceived as occupied by non-competing agents. Let the round end when either 1 or 0 agents compete for the resource. Corollary \ref{lm: hitting time of Y} states that in expectation this will require $\mathcal{O}\left( \frac{1}{p} \log N \right)$ steps.

	If all agents backed-off, it will take on average $R$ steps until at least one of them finds a free resource. We call this period a break.

	In the worst case, the system will oscillate between a round and a break. According to the above, one oscillation requires in expectation $\mathcal{O}\left( \frac{1}{p} \log N  + R\right)$ steps. If $R = 1$, Lemma \ref{lm: probability of entering state 1 Y} states that in expectation there will be $\frac{1}{\lambda} = \frac{2 - p}{2 (1 - p)}$ oscillations. For $R > 1$ the expected number of oscillations is bounded by $\mathcal{O}\left( R \frac{2 - p}{2 (1 - p)} \right)$. Thus, we derive the required bound (\ref{Eq: convergence time, constant p}).
\end{proof}

\subsubsection{Dynamic back-off Probability} \label{dynamic p}

So far we have assumed a constant back-off probability for each agent, i.e., $P_n(r, \prec_n) = p > 0, \forall n \in \mathcal{N}, \forall r \in \mathcal{R}$. In this section we will drop this assumption. Let $\psi = \max(\log N, R)$. Bound (\ref{Eq: convergence time, constant p}) becomes:

\begin{equation} \label{Eq: convergence time with psi}
	\mathcal{O}\left( \psi^2 \frac{2 - p}{2 (1 - p)} \left(\frac{1}{p} + 1 \right) \right) 
\end{equation}\normalsize

Intuitively, the worst case scenario corresponds to either all agents having a small back-off probability, thus they keep on competing for the same resource, or all of them having a high back-off probability, thus the process will keep on restarting. These two scenarios correspond to the inner ($\frac{1}{p}$) and outer ($\frac{2 - p}{2 (1 - p)}$) probability terms of bound (\ref{Eq: convergence time with psi}) respectively. We can rewrite the right part of bound (\ref{Eq: convergence time with psi}) as:

\begin{equation} \label{Eq: probabilities alternative form}
	\frac{2 - p}{2 (1 - p)} \left(\frac{1}{p} + 1 \right) = \frac{1}{p} + \frac{1}{1 - p} + \frac{1}{2} = \tau
\end{equation}\normalsize

As seen by Eq. \ref{Eq: probabilities alternative form}, $\tau$ assumes its maximum value on the two extremes, either with a high ($p \rightarrow 1^-$), or a low ($p \rightarrow 0^+$) back-off probability, i.e., $\underset{p \rightarrow 1^-}{\lim} \tau = \underset{p \rightarrow 0^+}{\lim} \tau = \infty$. Let $p^* = f(loss^*)$ be the worst between the smallest or highest back-off probability any agent $n \in \mathcal{N}$ can exhibit, i.e., having $loss^*$ given by Eq. \ref{Eq: loss*}. Using $p^*$ instead of the constant $p$, we bound the expected convergence time according to bound (\ref{Eq: convergence bound}).

\begin{equation} \label{Eq: convergence bound}
	\mathcal{O}\left( R \frac{2 - p^*}{2 (1 - p^*)} \left(\frac{1}{p^*} \log N + R \right) \right) 
\end{equation}\normalsize

\begin{equation} \label{Eq: loss*}
	loss^* = \underset{loss_n^r}{\argmin} \left( \underset{r \in \mathcal{R}, n \in \mathcal{N}}{\min}(loss_n^r), 1 - \underset{r \in \mathcal{R}, n \in \mathcal{N}}{\max}(loss_n^r) \right)
\end{equation}\normalsize

This concludes the proof of Theorem \ref{Th: convergence system}.

\end{document}